\documentclass[11pt]{article}

\usepackage{graphicx}
\usepackage{subcaption}
\usepackage[a4paper,top=2cm,bottom=2cm,left=3cm,right=3cm,marginparwidth=1.75cm]{geometry}
\usepackage{amsmath}
\usepackage{amssymb}
\usepackage{physics}
\usepackage{bm}
\usepackage{empheq}
\graphicspath{{images/}}
\usepackage{amsfonts}
\usepackage{epstopdf}
\usepackage{algorithmic}
\usepackage{amsopn}

\usepackage{amsthm}

\theoremstyle{plain}
\newtheorem{theorem}{Theorem}
\newtheorem{lemma}{Lemma}
\newtheorem{claim}{Claim}

\theoremstyle{definition}
\newtheorem{definition}{Definition}
\newtheorem{remark}{Remark}

\usepackage{cleveref}

\crefname{theorem}{theorem}{theorems}
\Crefname{theorem}{Theorem}{Theorems}

\crefname{lemma}{lemma}{lemmas}
\Crefname{lemma}{Lemma}{Lemmas}

\crefname{definition}{definition}{definitions}
\Crefname{definition}{Definition}{Definitions}

\crefname{remark}{remark}{remarks}
\Crefname{remark}{Remark}{Remarks}

\crefname{claim}{claim}{claims}
\Crefname{claim}{Claim}{Claims}

\title{Integral representation of time-harmonic solutions to Maxwell's equations with fast numerical convergence \thanks{This work was funded by the Air Force Office of Scientific Research under the grant no. FA9550-23-1-0093 and the AFOSR-MURI program FA9550-25-1-0262.}}

\author{Kalpesh Jaykar \thanks{Department of Aerospace Engineering and Mechanics, University of Minnesota, Minneapolis, MN (jayka001@umn.edu, james@umn.edu)} \and Richard D. James\footnotemark[2]}

\begin{document}

\maketitle

\pagenumbering{arabic}	

\begin{abstract}
    The robustness of XRD methods for the determination of the lattice parameters of crystals is well established. These methods have been extended to helical atomic structures using twisted x-rays \cite{friesecke_twisted_2016}. Building on an integral form
    used in \cite{friesecke_twisted_2016}, we construct integral representations of a broad class of time-harmonic solutions to Maxwell's equations in a vacuum or, more generally, in a homogeneous medium without source terms. The representation includes assignable generalized functions (distributions) that can be tailored to specific boundary or far-field conditions.  When the assignable functions satisfy mild periodicity and smoothness conditions, the solutions can be approximated using multi-dimensional trapezoidal rules with exponentially fast convergence. This approximation can be physically interpreted as utilizing finite sources of plane waves to approximate the broad class of time-harmonic solutions to Maxwell's equations. Using these solutions, we show that radiation from suitably placed and oriented sources can serve as incoming radiation for structures with icosahedral symmetry to achieve constructive interference after interacting with the icosahedral structure. The finite source approximations are sufficiently general to satisfy the general Dirichlet conditions at an arbitrarily large number of assigned locations in a source-free domain. The integral representation also extends to a broad class of physical phenomena governed by Helmholtz-type equations. Examples include the scalar wave equation for acoustic waves and elastic wave propagation in linear isotropic solids, which involve both scalar and vector wave equations. 
\end{abstract}



\section{Introduction}
Solving ordinary and partial differential equations (ODEs and PDEs) is central to modeling physical phenomena in engineering and physics.  Among the most important PDEs in electromagnetism are Maxwell's equations, which encapsulate the behavior of electric and magnetic fields and are fully consistent with the special theory of relativity.

A  foundational solution to Maxwell's equations is the plane wave: a solution in which the electric and magnetic fields are uniform across planes perpendicular to the direction of propagation with no free charges or currents in the medium \cite{Fry_PlaneWO_1927, Burrows_1965}. Plane waves are classical and widely studied solutions to Maxwell's equations in a vacuum or, more generally, in a homogeneous medium without sources. Normally, exploiting linearity,  we superpose plane waves traveling in orthonormal directions to approximate quite general
radiation, i.e., Fourier methods.  

But it is equally possible by linearity to superpose plane waves traveling in rather general non-orthonormal directions to obtain
solutions of Maxwell's equations. Such solutions are generally referred to as the angular spectrum \cite{morse_methods_1953, clemmow_plane_2013, Friberg_83}. These solutions can also equivalently be represented as a plane wave spectrum in terms of a field along a specific plane \cite{Hansen_Yaghijan_1999}.   
Typically, they are expressed in a specific Cartesian coordinate system and often rely on spatial Fourier transforms. Also, in some of these representations, orthogonality conditions may be required to satisfy the divergence-free constraint.

In this paper we generalize the latter by providing an explicit integral form of a solution of the
time-harmonic Maxwell's equations using  rotation matrices and arbitrary, assignable functions of the integration variables. Our formulation is fundamentally different from the Fourier-based approach and satisfies Maxwell's equations without imposing constraints on the assignable functions.  Also, the divergence-free
condition is automatically satisfied.
The explicit formulation and its potential applications are presented in \ref{sec:GeneralizationOfTwistedXrays}.

Our results generalize the integral form of twisted X-rays \cite{friesecke_twisted_2016}. Twisted X-rays are solutions to certain ``design equations" associated with helical symmetry, and they can be used as incoming radiation to determine the detailed atomic structures of molecules exhibiting helical symmetry \cite{friesecke_twisted_2016}. A summary of this work is provided in \Cref{sec:TwistedXRays}. Twisted X-rays can be interpreted as a superposition of infinitely many plane waves generated by rotating a plane wave around an oblique axis. Our integral form builds on this idea.   

While general superpositions of plane waves in arbitrary directions seem at first to nullify the main simplifying features of Fourier methods, other vast simplifications become possible.
Various numerical methods can be used to approximate our resulting integral form. If the integrand satisfies certain mild conditions, the computationally inexpensive trapezoidal rule demonstrates exponential convergence \cite{trefethen_exponentially_2014} compared to generic quadrature rules for approximating the integral. The fast convergence means that relatively  few plane waves are need to achieve 
acceptable error.  Moreover, the approximation has practical implications for the physical realization of the radiation. Since each term in the approximation is a classical plane wave, it is possible to generate the desired radiation, given by the integral formulation, using standard plane wave sources. This leads to a powerful consequence: any time-harmonic solution to Maxwell's equations in a source-free medium can, in principle, be approximated by the radiation of relatively few dipole antennas. We elaborate on this assertion in \Cref{sec:ApproximationOfGeneralization}.  

In summary, we present an integral form that  is
an exact solution of the time-harmonic 
Maxwell equations.  It exhibits extremely fast convergence when approximated by the trapezoidal rule. The terms in this approximation are themselves exact solutions of Maxwell's equations, interpretable as the radiation from synchronized dipole antennas. Thus, the time-harmonic Maxwell's equations are satisfied at 
both the discrete level and for the limit.
The discrete solution can be made to satisfy boundary conditions exactly at an
assignable finite number of boundary 
points.

Finally, since Maxwell's equations reduce to a 3D vector Helmholtz equation, our approach naturally extends to other physical systems governed by the Helmholtz equations. This includes wave propagation problems in 3D elasticity or spherical scalar waves governed by scalar Helmholtz equations. These extensions are discussed in \Cref{sec:OtherRelevantProblem}.

\section{Background : Design equations for twisted X-rays}\label{sec:TwistedXRays}

We begin by reviewing some results from the work of Friesecke et al.\cite{friesecke_twisted_2016}, which form the foundation for the development of this paper. The differential form of Maxwell's equations in vacuum is: 

\begin{align}\label{eq:MaxwellEquationInFreeSpace}
    \div{\mathbf{E}} = 0, \quad
   \div{\mathbf{B}} = 0, \quad
   \curl{\mathbf{E}} = -\frac{\partial \mathbf{B}}{\partial t}, \quad
   \curl{\mathbf{B}} = \frac{1}{c^2} \frac{\partial \mathbf{E}}{\partial t}.
\end{align}

Assuming time-harmonic (monochromatic) fields of the form $\mathbf{E}(\mathbf{x},t) = \mathbf{E}_0(\mathbf{x})e^{-i\omega t}$, $\mathbf{B}(\mathbf{x},t) = \mathbf{B}_0(\mathbf{x})e^{-i\omega t}$, for some $\omega > 0$, the differential equations for the electric and magnetic fields decouple and reduce to:

\begin{subequations}\label{eq:HelmholtzEquation}
\begin{align}
    \Delta \mathbf{E}_0 &= -\frac{\omega^2}{c^2} \mathbf{E}_0, \qquad \div{\mathbf{E}_0} = 0, \label{eq:ElectricField} \\
    \mathbf{B}_0 & = -\frac{i}{\omega} \nabla \times \mathbf{E}_0.\label{eq:MagneticField}
\end{align}
\end{subequations}

Thus the electric field $\mathbf{E}_0$ satisfies the vector Helmholtz equations and divergence-free condition, and the magnetic field $\mathbf{B}_0$ is determined from $\mathbf{E}_0$. 

As discussed earlier, plane waves are foundational solutions to Maxwell's equations and are widely used in conventional X-ray diffraction for analyzing crystal structures. In \cite{friesecke_twisted_2016}, the authors established a connection between plane waves and the symmetry of crystal structures. They show that constructive interference arises when incoming radiation shares the symmetry of crystals, whereas a larger, continuous symmetry of the incoming radiation leads to destructive interference. This relationship is formulated using group-theoretic language: to produce a desired interference pattern, incoming radiation should be a solution to Maxwell's equations, a divergence-free eigenfunction of the Laplacian operator, that is also an eigenfunction of the continuous extension of the group generating the structure.

Let $G_0$ denote the discrete symmetry group associated with the structure — for example, the discrete translation group in the case of a crystal. Let $G$ be the continuous extension of $G_0$ and let $g \in G$. The ``design equations" for the desired incoming radiation are:

\begin{align}
    (g\mathbf{E}_0)(\mathbf{x}) &= \chi_g \mathbf{E}_0(\mathbf{x}),& \label{eq:DesignEquation1}\\
    \Delta \mathbf{E}_0 (\mathbf{x}) &= -\frac{\omega^2}{c^2} \mathbf{E}_0(\mathbf{x}), \quad & \nabla\cdot\mathbf{E}_0(\mathbf{x}) = 0.\label{eq:DesignEquation2}
\end{align}

Here, \Cref{eq:DesignEquation1} states that $\mathbf{E}_0$ is an eigenfunction of the group, with eigenvalue $\chi_g$, and \Cref{eq:DesignEquation2} ensures that $\mathbf{E}_0$ satisfies Maxwell's equations. To interpret the group action 
$g\mathbf{E}_0$, we first define the notation $g=(\mathbf{Q}|\mathbf{c})$ for an element of the Euclidean isometry group, acting on $\mathbf{x} \in \mathbb{R}^3$ by
\begin{align*}
    (\mathbf{Q}|\mathbf{c})\mathbf{x} = \mathbf{Q}\mathbf{x}+\mathbf{c},
\end{align*}
where $\mathbf{Q}$ is orthogonal rotation matrix and $\mathbf{c}\in\mathbb{R}^3$ is a translation vector. The natural action on a vector fields $\mathbf{E_0} : \mathbb{R}^3 \to \mathbb{C}^3$ is $ \left( \left(\mathbf{Q}|\mathbf{c}\right) \mathbf{E}_0\right) \left(\mathbf{x}\right) = \mathbf{Q} \mathbf{E}_0\left( \mathbf{Q}^{-1}\left( \mathbf{x} - \mathbf{c}\right) \right)$. In \cite{friesecke_twisted_2016}, the group $G$ is taken to be the helical group $\mathcal{H}_{\bf{e}}  = \{ (\bf{R}_\theta|\tau\bf{e}) : \theta\in [0,2\pi), \tau \in \mathbb{R} \}$, where $ \bf{R}_{\theta} $ denotes a rotation about the axis $\mathbf{e}$, and $\tau\mathbf{e}$ is a translation along that axis. This is analogous to a screw displacement, and the group product is: 

\begin{align*}
(\mathbf{R}_{\theta_1}|\tau_1\mathbf{e})(\mathbf{R}_{\theta_2}|\tau_2\mathbf{e}) = (\mathbf{R}_{\theta_1 + \theta_2}|(\tau_1+\tau_2)\mathbf{e}).
\end{align*}
The action of this group on the electric field, expressed in cylindrical coordinates $(r,\phi,z)$ with the z-axis aligned with $\mathbf{e}$, is defined as:
    \begin{align*}
        ((\mathbf{R}_{\theta}|\tau\mathbf{e})\mathbf{E}_0)(\mathbf{x}(r,\phi,z)) = \mathbf{R}_\theta\mathbf{E}_0(\mathbf{R}_\theta^T \mathbf{x}(r,\phi,z)-\tau\mathbf{e}) = \mathbf{R}_\theta\mathbf{E}_0(\mathbf{x}(r,\phi-\theta,z-\tau)).
    \end{align*}

The explicit solution to the design equations for the helical group, derived in \cite{friesecke_twisted_2016}, is: 
    \begin{equation*}
        \mathbf{E}_0  (r,\phi,z) = e^{i(\alpha\phi+\beta z)} \bf{R}_\phi \bf{N}(\bf{n}) \begin{pmatrix}  J_{\alpha+1}(\gamma r)  \\ J_{\alpha-1}(\gamma r)  \\ J_{\alpha}(\gamma r) \end{pmatrix} ,
    \end{equation*} 
where $(\alpha,\beta,\gamma)\in\mathbb{Z}\times\mathbb{R}\times(0,\infty)$, $\mathbf{n}\in \mathbb{C}^3$, and the constraint $(0,\gamma,\beta)\cdot\mathbf{n}=0$ ensures the divergence-free condition. Here, $J_\alpha$ denotes the Bessel function of the first kind of order $\alpha$ and $\bf{N}(\bf{n})$ is a $3 \times 3$ matrix depending linearly on $\mathbf{n} \equiv \left(n_1,n_2,n_3\right)$:
    \begin{equation*}
        \bf{N}(\bf{n})= \begin{pmatrix}  \frac{n_1 + i n_2}{2} &\frac{n_1 - i n_2}{2} & 0 \\ \frac{n_2 - i n_1}{2}&\frac{n_2 + i n_1}{2}&0  \\ 0&0&n_3   \end{pmatrix},
    \end{equation*}
The associated frequency of the time-harmonic fields is $\omega = c |(0,\gamma,\beta)|$. Using the integral representation of Bessel functions and appropriate algebraic manipulations, the solution can also be written in the integral form: 

    \begin{equation}\label{eq:twistedXrays}
        \mathbf{{E_0}} (\mathbf{x}) = \frac{1}{2\pi} \int_{-\pi}^{\pi} e^{i\alpha\psi} (\mathbf{R_\psi} \mathbf{n}) e^{i(\mathbf{R_\psi} \mathbf{k}\cdot\bf{x})} \mathrm{d}\psi,
    \end{equation}
where $\mathbf{k} = \left(0, \gamma, \beta\right)$ is a wave vector, and $\mathbf{R}_\psi$ denotes a rotation by angle $\psi$ about the helical axis. The integrand represents a plane wave propagating in the direction $\mathbf{R}_\psi \mathbf{k}$ with polarization $\mathbf{R}_\psi \mathbf{n}$, and phase factor $e^{i \alpha \psi}$. This representation shows that the twisted X-ray field is a continuous superposition of rotated plane waves, with the axis of rotation parallel to the helical axis.

\section{Generalization of twisted X-rays}\label{sec:GeneralizationOfTwistedXrays}

Although twisted X-rays were originally introduced as solutions to the design equations associated with the helical isometry group, their underlying structure — a class of solutions to Maxwell's equations constructed as a superposition of plane waves — can be very useful for solving design equations for other groups and suggests a broader applicability. This naturally raises several questions: Can the principle of superposing plane waves, as seen in twisted X-rays, be extended to more general solutions to Maxwell's equations? What is the most general form of time-harmonic solutions to Maxwell's equations that can be constructed using plane waves as a basis? Does such a construction span the entire space of time-harmonic solutions to Maxwell's equations in free space, or are there any solutions that lie outside this solution space? By a ``general" solution, we refer to a family of solutions parameterized by sufficient degrees of freedom to represent all solutions to the governing differential equations — in this case, Maxwell's equations in free space. In this section, several of these questions will be addressed.

We begin analyzing the integral form of the twisted X-rays given in \Cref{eq:twistedXrays}. The spatial dependence of the solution appears solely in the exponential term of the integrand, $e^{i \mathbf{R}_\psi \mathbf{k} \cdot \mathbf{x}}$, which is an eigenfunction of the Laplacian operator. This observation suggests that plane waves of the form $e^{i \mathbf{k}\cdot\mathbf{x}}$ can serve as a basis for the solutions to the Helmholtz equation — the first equation in \Cref{eq:DesignEquation2} — provided that $\mathbf{k} \cdot \mathbf{k} = \omega^2/c^2$. Since rotations preserve inner products, $\mathbf{R}_\psi \mathbf{k} \cdot \mathbf{R}_\psi \mathbf{k} = \mathbf{k} \cdot \mathbf{k} $, plane waves with rotated wave vectors also satisfy the Helmholtz equation. Therefore, any solution of the form $\mathbf{E}_0(\mathbf{x}) = \sum_{\psi} \mathbf{a}(\psi) e^{i \mathbf{R}_\psi\mathbf{k}\cdot\mathbf{x}}$ or its continuous analogue $\mathbf{E}_0(\mathbf{x}) = \int \mathbf{a}(\psi) e^{i \mathbf{R}_\psi \mathbf{k} \cdot \mathbf{x}} \mathrm{d}\psi$, as seen in the twisted X-rays, satisfies the first equation in \Cref{eq:DesignEquation2}. This idea is generalized in the subsequent construction. The second equation in \Cref{eq:DesignEquation2} imposes the divergence-free condition, which requires that $\mathbf{R}_\psi \mathbf{k} \cdot \mathbf{a}(\psi) = 0$. This condition suggests that the polarization or amplitude vector $\mathbf{a}(\psi)$ must be orthogonal to the corresponding wave vector $\mathbf{R}_\psi \mathbf{k}$. In twisted X-rays, this holds automatically because $\mathbf{a}(\psi) = \mathbf{R}_\psi \mathbf{n}$ and $\mathbf{R}_\psi \mathbf{k} \cdot \mathbf{R}_\psi \mathbf{n} = \mathbf{k} \cdot \mathbf{n} = 0$. We retain this orthogonality requirement in the general case.

Next, we consider a superposition of time-harmonic plane waves propagating in all directions, allowing full freedom in specifying their amplitudes and phases while ensuring that the polarization vectors remain orthogonal to the corresponding wave vectors. Each constituent plane wave satisfies Maxwell's equations and shares a common time-harmonic dependence of the form: $$\mathbf{E}(\mathbf{x},t) = \mathbf{E}_0 (\mathbf{x})e^{-i \omega t}, \quad \mathbf{B}(\mathbf{x},t) = \mathbf{B}_0 (\mathbf{x})e^{-i \omega t}, \quad \omega>0.$$ By linearity, the superposition of such plane waves also retains the time-harmonic form, allowing us to focus exclusively on the spatial components. As shown in \Cref{eq:MagneticField}, the magnetic field is determined by the electric field via: $$\mathbf{B}_0(\mathbf{x}) = -\frac{i}{\omega} \nabla \times \mathbf{E}_0(\mathbf{x}).$$ Thus, it suffices to construct $\mathbf{E}_0 (\mathbf{x})$ alone by solving \Cref{eq:ElectricField}.

To construct a superposition of plane waves, consider a standard plane wave $\mathbf{n}_0 e^{i \mathbf{k}_0\cdot \mathbf{x}}$, where $\mathbf{k}_0 \cdot \mathbf{k}_0 = \omega^2/c^2$ with polarization $\mathbf{n}_0 \cdot \mathbf{k}_0=0$. Rotating the polarization about $\mathbf{k}_0$ generates a one-parameter family $\mathbf{R}_1 (\theta_1) \mathbf{n}_0 e^{i \mathbf{k}_0 \cdot \mathbf{x}}$, where $\mathbf{R}_1 (\theta_1)$ is a rotation about wave vector $\mathbf{k_0}$ and $\theta_1$ is a free parameter. Assuming $\mathbf{k}_0 \in \mathbb{R}^3$, it has two degrees of directional freedom and we introduce rotations $\mathbf{R}_2 (\theta_2)$ and $\mathbf{R}_3 (\theta_3)$ about two mutually orthogonal axes perpendicular to $\mathbf{k}_0$. Thus, a general three-parameter family of plane waves is: $$ \mathbf{R}_3(\theta_3) \mathbf{R}_2(\theta_2) \mathbf{R}_1(\theta_1) \mathbf{n}_0 e^{i \mathbf{R}_3(\theta_3) \mathbf{R}_2(\theta_2) \mathbf{k}_0 \cdot \mathbf{x}}, $$ where $\theta_2$,$\theta_3$ are free parameters. A more general case with $\mathbf{k}_0 \in \mathbb{C}^3$, relevant for modeling electromagnetic radiation in a dissipative medium or evanescent waves, is discussed in \Cref{sec:ComplexWaveVector}. A general superposition with a complex amplitude $\hat{A}(\theta_1,\theta_2,\theta_3)$ that encodes both the magnitude and relative phase of the components, can be written as:

\begin{align*}
    \mathbf{E}_0(\mathbf{x}) &= \int_{-\pi}^\pi \int_{-\pi/2}^{\pi/2} \int_{-\pi}^\pi \hat{A}(\theta_1,\theta_2,\theta_3) \mathbf{R}_3 (\theta_3) \mathbf{R}_2 (\theta_2) \mathbf{R}_1 (\theta_1) \mathbf{n}_0 e^{i (\mathbf{R}_3 (\theta_3) \mathbf{R}_2 (\theta_2) \mathbf{k}_0 \cdot \mathbf{x})} \mathrm{d}\theta_1 \mathrm{d}\theta_2 \mathrm{d}\theta_3.
\end{align*}

When the propagation direction is fixed, superposing over $\theta_1$ changes the polarization but not the propagation direction. Therefore, we can absorb the $\theta_1$-dependence into the amplitude, in which case the field becomes:

\begin{align*}
    \mathbf{E}_0(\mathbf{x}) &= \int_{-\pi}^\pi \int_{-\pi/2}^{\pi/2} A(\theta_2,\theta_3) \mathbf{R}_3 (\theta_3) \mathbf{R}_2 (\theta_2) \mathbf{R}_1\!\left(\theta_1(\theta_2,\theta_3)\right) \mathbf{n}_0 e^{i (\mathbf{R}_3 (\theta_3) \mathbf{R}_2 (\theta_2) \mathbf{k}_0 \cdot \mathbf{x})} \mathrm{d}\theta_2 \mathrm{d}\theta_3.
\end{align*}
In this form, the function $\theta_1 ( \theta_2, \theta_3 )$ represents the dependence of polarization vector on the propagation direction. The resulting time-harmonic electric field $\mathbf{E}(\mathbf{x},t) = \mathbf{E}_0 (\mathbf{x}) e^{-i \omega t}$ is then given by:

\begin{align} \label{eq:FinalSolution}
    \mathbf{E}(\mathbf{x},t) &= e^{-i \omega t} \int_{-\pi}^\pi \int_{-\pi/2}^{\pi/2} A(\theta_2,\theta_3) \mathbf{R}_3 (\theta_3) \mathbf{R}_2 (\theta_2) \mathbf{R}_1\! \left(\theta_1(\theta_2,\theta_3)\right) \mathbf{n}_0 e^{i (\mathbf{R}_3 (\theta_3) \mathbf{R}_2 (\theta_2) \mathbf{k}_0 \cdot \mathbf{x})} \mathrm{d}\theta_2 \mathrm{d}\theta_3
\end{align}

This representation admits the following physical interpretation: a time-harmonic electromagnetic field in free space is represented as a continuous superposition of plane waves, each with a specific propagation direction and polarization. The complex amplitude function $A(\theta_2,\theta_3)$ encodes the contribution of each wave component, while the rotation matrices and the function $\theta_1 (\theta_2, \theta_3)$ define the orientations of the wavevector and polarization, respectively. This formulation extends the twisted X-ray construction by enabling arbitrary control over the angular distribution of energy and polarization. As a result, it is suitable for broader applications.

We claim that this representation spans all bounded time-harmonic solutions of Maxwell's equations in free space \eqref{eq:MaxwellEquationInFreeSpace}, provided that the amplitude function $A(\theta_2,\theta_3)$, supported on $\mathbb{S}^2$, is allowed to be a distribution. Evanescent waves, as well as waves in dissipative media that require complex wave vectors, are excluded from this representation and are discussed separately in \Cref{sec:ComplexWaveVector}.

To support this claim, we begin by defining the function spaces used for the electric field. Specifically, we consider the space of locally integrable functions of polynomial growth (also known as slowly increasing or tempered functions). Let
\begin{itemize}
    \item $\mathcal{A}(\mathbb{R}^3)$ be the space of functions $f:\mathbb{R}^3 \to \mathbb{C}$ such that $f \in L^1_{loc} (\mathbb{R}^3) $ and there exists an integer $k\geq0$ satisfying $\int_{\mathbb{R}^3} \frac{|f(\mathbf{x})|}{(1+|x|^2)^k} \mathrm{d}\mathbf{x} < \infty$.
    \item $\mathcal{A}^3$ be the space of vector-valued function $\mathbf{F}= (f_1,f_2,f_3)$, with each component $f_i \in \mathcal{A}$.
    \item $\mathcal{S}(\mathbb{R}^3)$ denote the Schwartz space of rapidly decreasing smooth functions.
    \item $\mathcal{S}'(\mathbb{R}^3)$ denotes the space of tempered distributions, (continuous dual space of $\mathcal{S}$).
    \item $\mathcal{E}(\mathbb{R}^3)$ denotes the space $C^\infty(\mathbb{R}^3)$ of smooth functions,
    \item $\mathcal{E}'(\mathbb{R}^3)$ denotes the space of compactly supported distributions (continuous dual space of $\mathcal{E}$).
\end{itemize}

Any function in $\mathcal{A}$ induces a tempered distribution, $\mathcal{A}^3 \subset \mathcal{S}'^3$ (see, e.g., Example 1.11.2 in \cite{kesavan_topics_1989}). Because $\mathcal{A}^3 \subset \mathcal{S}'^3$, operators such as $\nabla\cdot\mathbf{E}_0$ (divergence) and $\nabla\times \mathbf{E}_0$ (curl) are well-defined in the distributional sense for every $\mathbf{E}_0\in \mathcal{A}^3$.

We briefly recall the basic notions from distribution theory that will be used throughout this section. 

\begin{definition} \label{def:ActionOfDistribution}
    The action of a distribution $\mathcal{T} \in \mathcal{E}'$ on a test function $f\in\mathcal{E}$ is a scalar, denoted by $\mathcal{T}(f)$ or, more commonly, by the dual pairing bracket $\left< \mathcal{T}, f\right>$. 
\end{definition}

\begin{remark}\label{rem:Distributions}
    If a distribution $\mathcal{T} \in \mathcal{E}'(\mathbb{R}^3)$ is non-singular, it can be identified with a smooth function $T$, and a well-defined integral gives its action:
    $\left< \mathcal{T}, f \right> = \int_{\mathbb{R}^3} T(\mathbf{x}) f(\mathbf{x}) \mathrm{d}\mathbf{x}.$
    For singular distributions (e.g., the Dirac delta), such a function $T$ does not exist. Nevertheless, we use integral notation symbolically to represent distributional action $\left< \mathcal{T}, f\right>$. Note that, non-singular distributions are dense in all distributions (see, e.g., Theorem 4.1.5 and 7.1.8 in \cite{hormander_analysis_1990}).
\end{remark} 

We now recall a standard result relating the support of a distribution to multiplication by smooth functions.

\begin{lemma} \label{lemma:DistributionWithCompactSupport}
    If the product of a polynomial $p$ and a distribution $\mathcal{T} \in \mathcal{S}'$ is the zero distribution, then the support of $\mathcal{T}$ is contained in the zero set of $p$.
\end{lemma}

\begin{proof}
    Since $p$ is a polynomial and hence smooth, the zero set of $p$ is closed. Let $f\in\mathcal{S}$ be any test function. Then $p\mathcal{T}(f) = \mathcal{T} (p f ) =0$ from the definition of multiplication of a function by a distribution. Now, any test function $g\in \mathcal{S}$ supported outside the zero set of $p$, can be written as $g= p h$ with $h\in \mathcal{S}$ and the support of $g$ is contained in the open set complementary to the zero set of $p$. Thus, $\mathcal{T}(g) = \mathcal{T} (p h) =0$. 
\end{proof}

\begin{remark}\label{rem:CompactDistributions}    
    Any distribution with compact support necessarily has finite order (see corollary of Theorem 3.1.2 in \cite{Friedlander1998}). The order of a distribution $\mathcal{T} \in \mathcal{E}$ is the smallest integer $N$ such that, for every compact set $K\subset \mathbb{R}^3$, there exists a constant $C$ such that, $\left< \mathcal{T}, f\right> \leq C \sup_{\left| \alpha \right| \leq N}{\left| \partial^\alpha f \right|}$, for all $f\in\mathcal{E}(\mathbb{R}^3)$ supported in $K$.

    In particular, distributions $\mathcal{T}$ supported on $\mathbb{S}^2$ and of order zero (regular distributions) admit the representation $$\left< \mathcal{T}, f \right> = \int_{\mathbb{S}^2} T(\mathbf{x}) f(\mathbf{x}) \mathrm{d}\mathbf{x},$$ for some integrable function $T$ defined on the sphere and for all $f\in\mathcal{E}(\mathbb{R}^3)$. Here, $\mathrm{d}\mathbf{x}$ denotes the induced (Lebesgue) surface measure on the $\mathbb{S}^2$.
\end{remark}

\begin{theorem}
    Let $\mathbf{E}_0 (\mathbf{x})\in\mathcal{A}^3(\mathbb{R}^3)$ be the spatial part of a time-harmonic solution of Maxwell's equations in free space \eqref{eq:MaxwellEquationInFreeSpace}. Then $\mathbf{E}(\mathbf{x},t) = \mathbf{E}_0(\mathbf{x}) e^{-i\omega t} $ is smooth and can be represented in the form \eqref{eq:FinalSolution} with appropriate distributions $A(\theta_1, \theta_2)$ and $\theta_1(\theta_2,\theta_3)$.
\end{theorem}

\begin{proof}
     We consider solutions $\mathbf{E}_0 \in \mathcal{A}^3$, i.e., $\mathbf{E}_0=(E_1,E_2,E_3)$ with each $E_i\in\mathcal{A}\subset\mathcal{S}'$. The Fourier transform of the $E_i$ exists and lies in $\mathcal{S}'$. Note that, applying the Fourier transform twice yields $\mathcal{F} \left( \mathcal{F} \left( E_i \right) \right) (\mathbf{x}) = E_i (-\mathbf{x})$. Since $E_i$ satisfies the Helmholtz equation, we obtain:
     \begin{align*}
         \left( \mathbf{k}\cdot\mathbf{k} - \frac{\omega^2}{c^2}\right) \hat{E_i}(\mathbf{k})=\mathbf{0}.
     \end{align*}
     By \Cref{lemma:DistributionWithCompactSupport}, each $\hat{E_i}$ is supported on the sphere $\mathbf{k}\cdot\mathbf{k} = \omega^2/c^2$ and hence it can be viewed as $\mathcal{E}' \left( \left( \omega/c \right)\mathbb{S}^2\right)$, continuous dual space of $C^\infty\left( \left( \omega/c \right)\mathbb{S}^2\right)$. Thus, $\hat{E_i} \in\mathcal{E}' \subset \mathcal{S}'$ and hence $\hat{E_i}$ is a distribution of finite order (see \Cref{rem:CompactDistributions}). For $f\in\mathcal{E}$,
     \begin{align*}
         \hat{E_i}(f) = \frac{1}{(2\pi)^3}\int_{\mathbf{k}\cdot\mathbf{k} = \omega^2/c^2} N_i(\mathbf{k}) f({\mathbf{k}}) \mathrm{d}\mathbf{k},
     \end{align*}
     for some Lebesgue integrable functions on the sphere $N_i$ (see \Cref{rem:Distributions}). By the Paley-Wiener-Schwartz theorem, since $\hat{E_i}$ is a distribution with compact support, its Fourier transform is a smooth function with tempered growth, which can be expressed by the action of $\hat{E_i}$ on the exponential basis \cite{hormander_analysis_1990}:
     \begin{align*}
         \mathcal{F} (\hat{E_i}) (\mathbf{x}) = \hat{E}_{i} (e^{-i\mathbf{k}\cdot\mathbf{x}}).
     \end{align*}
     Note that $\hat{E}_{i}$ acts on $e^{-i\mathbf{k}\cdot\mathbf{x}}$ as a function of $\mathbf{k}$. We can recover the $E_i$ as follows:

     \begin{align*}
         E_i(\mathbf{x}) &=  (2\pi)^3 \mathcal{F} (\hat{E_i}) (-\mathbf{x})\\
         &= (2\pi)^3 \hat{E}_{i} (e^{i\mathbf{k}\cdot\mathbf{x}})\\
         &= \int_{\mathbf{k}\cdot\mathbf{k} = \omega^2/c^2} N_i(\mathbf{k}) e^{i\mathbf{k}\cdot\mathbf{x}} \mathrm{d}\mathbf{k}.
     \end{align*}
     The complete spatial part of the electric field can be written as 
     \begin{align*}
         \mathbf{E}_0(\mathbf{x}) &= \int_{\mathbf{k}\cdot\mathbf{k} = \omega^2/c^2} \mathbf{N}(\mathbf{k}) e^{i\mathbf{k}\cdot\mathbf{x}} \mathrm{d}\mathbf{k},
     \end{align*}
     for some Lebesgue integrable vector functions on the sphere  $\mathbf{N}(\mathbf{k})$. We redefine the variable $\mathbf{k} = (\omega/c) \hat{\mathbf{k}}$, where $\hat{\mathbf{k}}\in \mathbb{S}^2$. We can parameterize $\mathbb{S}^2$ by two parameters. Let $\hat{\mathbf{k}}_0$ and $\mathbf{n}_0$ be any two orthogonal vectors such that $|\hat{\mathbf{k}}_0|=|\mathbf{n}_0|=1$ and $\{ \hat{\mathbf{k}}_0,\mathbf{n}_0 \times \hat{\mathbf{k}}_0,\mathbf{n}_0 \} $ forms an orthonormal basis in $\mathbb{R}^3$. Define SO(3)-valued functions $\mathbf{R}_2(\theta_2)$ and $\mathbf{R}_3(\theta_3)$ such that $\mathbf{R}_2 (\theta_2) (\mathbf{n}_0 \times \hat{\mathbf{k}}_0) = (\mathbf{n}_0 \times \hat{\mathbf{k}}_0)$ for any $\theta_2 \in [-\pi/2,\pi/2]$ and $\mathbf{R}_3 (\theta_3) \mathbf{n}_0 = \mathbf{n}_0$ for any $\theta_3 \in [-\pi,\pi)$. We can see that $\mathbf{R}_3(\theta_3) \mathbf{R}_2 (\theta_2)\hat{\mathbf{k}}_0$, such that $\theta_2 \in [-\pi/2,\pi/2]$ and $\theta_3 \in [-\pi,\pi)$, spans  $\mathbb{S}^2$, from the geometrical interpretation of $\theta_2, \theta_3$ as being polar and azimuthal angles in spherical coordinates. The solution $\mathbf{E}_0$ takes the form:

    \begin{align*}
        \mathbf{E}_0(\mathbf{x}) &= \int_{\mathbb{S}^2} \mathbf{N}(\hat{\mathbf{k}}) e^{i (\omega/c)\hat{\mathbf{k}}\cdot\mathbf{x}} \mathrm{d}\hat{\mathbf{k}}\\
        &=\int_{\theta_3=-\pi}^{\pi} \int_{\theta_2=-\pi/2}^{\pi/2} \mathbf{N}(\theta_3,\theta_2) e^{i\mathbf{R}_3(\theta_3) \mathbf{R}_2(\theta_2)(\omega/c)\hat{\mathbf{k}}_0\cdot\mathbf{x}} \mathrm{d}\theta_2 \mathrm{d}\theta_3
    \end{align*}
    The representation is constructed using Helmholtz equations, and thus we can verify that it satisfies  \Cref{eq:HelmholtzEquation}. However, the divergence-free condition imposes additional restrictions on the orthogonality of $\mathbf{N}(\theta_2 ,\theta_3)$ and $\mathbf{R}_3(\theta_3) \mathbf{R}_2(\theta_2)(\omega/c)\hat{\mathbf{k}}_0$.

    \begin{align*}
        \nabla\cdot\mathbf{E}_0 (\mathbf{x}) &=\int_{\theta_3=-\pi}^{\pi} \int_{\theta_2=-\pi/2}^{\pi/2} \mathbf{R}_3(\theta_3) \mathbf{R}_2(\theta_2)(\omega/c)\hat{\mathbf{k}}_0\cdot\mathbf{N}(\theta_3,\theta_2) e^{i\mathbf{R}_3(\theta_3) \mathbf{R}_2(\theta_2)(\omega/c)\hat{\mathbf{k}}_0\cdot\mathbf{x}} \mathrm{d}\theta_2 \mathrm{d}\theta_3.
    \end{align*}
    Thus, we seek most general $\mathbf{N}(\theta_2 ,\theta_3)$ such that, $\mathbf{N}(\theta_2 ,\theta_3) \cdot \mathbf{R}_3(\theta_3) \mathbf{R}_2 (\theta_2) (\omega / c) \hat{\mathbf{k}}_0 = 0$ for all $\theta_2$ and $\theta_3$. This condition is satisfied if $\mathbf{R}_2(-\theta_2) \mathbf{R}_3(-\theta_3) \mathbf{N}(\theta_2 ,\theta_3)$ lies in the plane perpendicular to $\hat{\mathbf{k}}_0$, which is the plane spanned by $(\mathbf{n}_0 \times \hat{\mathbf{k}}_0)$ and $\mathbf{n}_0$. Define  $\mathbf{R}_1(\theta_1)$ such that $\mathbf{R}_1 (\theta_1) \hat{\mathbf{k}}_0 = \hat{\mathbf{k}}_0$ for any $\theta_1 \in [0,2\pi)$. Then, $\mathbf{R}_1 (\theta_1) \mathbf{n}_0$ spans the entire plane orthogonal to $\hat{\mathbf{k}}_0$. Set $ \mathbf{N}(\theta_2 ,\theta_3) = A(\theta_2, \theta_3) \mathbf{R}_3(\theta_3) \mathbf{R}_2(\theta_2) \mathbf{R}_1 \left(\theta_1 (\theta_2,\theta_3) \right) \mathbf{n}_0$, for any $\theta_1 \left( \theta_2,\theta_3 \right) \in [0,2\pi)$. Defining $\mathbf{k}_0 = (\omega/c) \hat{\mathbf{k}}_0$, we obtain the final form for the smooth time harmonic solution to Maxwell's equations in a free space as:

    \begin{align*}
        \mathbf{E}_0 (\mathbf{x}) &= \int_{\theta_3=-\pi}^{\pi} \int_{\theta_2=-\pi/2}^{\pi/2} A(\theta_2, \theta_3) \mathbf{R}_3(\theta_3) \mathbf{R}_2(\theta_2) \mathbf{R}_1 \left(\theta_1 (\theta_2,\theta_3) \right) \mathbf{n}_0  e^{i \mathbf{R}_3(\theta_3) \mathbf{R}_2 (\theta_2) \mathbf{k}_0 \cdot \mathbf{x}} d\theta_2 d\theta_3
    \end{align*}

    In the case of singular distribution, $A(\theta_2,\theta_3)$ and/or $\theta_1(\theta_2,\theta_3)$ represent distributions supported on sphere $\mathcal{E}' \left( \left( \omega/c \right)\mathbb{S}^2\right)$, and the integration is symbolic.
\end{proof}

Although the above representation is derived for a source-free vacuum medium, it also applies to any source-free non-dissipative homogeneous medium after replacing the propagation velocity $c$ by the appropriate velocity. The field $\mathbf{E}_0:\mathbb{R}^3\to\mathbb{C}^3$ is defined on all of $\mathbb{R}^3$, and the requirement $\mathbf{E}_0\in\mathcal{A}^3$ restricts the solution space to bounded electric fields on all of $\mathbb{R}^3$. In particular, this representation does not include electric fields associated with complex wavevectors. Such fields—corresponding to evanescent waves—are treated separately in \Cref{sec:ComplexWaveVector}.

\begin{remark}
From a physical perspective, the representation \eqref{eq:FinalSolution} with $A,\theta_1\in\mathcal{E}(\mathbb{S}^2)$ (regular distribution) can be interpreted as a superposition of plane waves propagating in all directions $(\theta_2,\theta_3)$, thereby defining a tangential field on $\mathbb{S}^2$. If $A$ and $\theta_1$ are continuous on $\mathbb{S}^2$, then the associated tangential field must vanish at least at one point on the sphere, by the Hairy Ball theorem.
\end{remark}

\subsection{Application: Reflection from a perfect cylinder}
To illustrate how the general representation developed in this section can be applied to physically relevant scattering problems, we consider the reflection of an electromagnetic wave from an infinitely long, perfectly conducting cylinder of radius $\gamma$. This setting naturally aligns with the integral representation in \Cref{eq:FinalSolution} and demonstrates how reflected fields can be constructed explicitly by the appropriate choice of the amplitude distribution.

Let the incoming radiation be the linearly polarized plane wave $$\mathbf{E}_{inc}(\mathbf{x}) = \mathbf{e}_2 e^{i|\mathbf{k}| (\mathbf{e}_1 \cdot \mathbf{x})}, $$ where 
\begin{itemize}
    \item $\mathbf{e}_1$ is the propagation direction,
    \item $\mathbf{e}_2$ is the polarization diretion with $\mathbf{e}_2\cdot\mathbf{e}_1=0$,
    \item $|\mathbf{k}|=\omega/c$,
    \item $\mathbf{e}_3=\mathbf{e}_1 \times \mathbf{e}_2$ is the cylinder axis.
\end{itemize}
An ideal conductor produces a phase shift of $\pi$ upon reflection \cite{Born_Wolf_1999}. The reflected field is obtained by choosing a suitable amplitude function in \Cref{eq:FinalSolution}. To represent the reflected field , we select the amplitude distribution  $$A(\theta_2,\theta_3)= 0.5 \delta(\theta_2) e^{i \left| \mathbf{k} \right| ( r_0 - 2 \gamma \cos{\left(\theta_3/2\right)} ) } \cos{\left(\theta_3/2\right)} U\left(\mathbf{x}\cdot\left( \cos{\left(\theta_3/2 \right)} \mathbf{e}_1 + \sin{\left(\theta_3/2 \right)}\mathbf{e}_2 \right) - \gamma\right)$$ and set $$\theta_1 \left( \theta_2, \theta_3\right) = -\pi/2.$$ Here: 
\begin{itemize}
    \item $\delta(\theta_2)$ selects wavevectors lying in the plane of incidence.
    \item The factor $\cos(\theta_3/2)$ represents the geometric projection of incident flux onto the curved surface.
    \item $U(\cdot)$ is the Heaviside unit step function ensuring the field vanishes inside the conductor.
    \item $r_0$ is an arbitrary reference distance used only to fix the global phase of the reflected field. 
\end{itemize}
The rotations in \Cref{eq:FinalSolution} are chosen to align with the problem geometry:
\begin{align*}
    \mathbf{R}_1(\theta_1) \mathbf{e_1} = \mathbf{e}_1, \quad
    \mathbf{R}_2(\theta_2) \mathbf{e_2} = \mathbf{e}_2, \quad
    \mathbf{R}_3(\theta_3) (\mathbf{e}_3) = \mathbf{e}_3.
\end{align*}
Let $\theta_3/2 =\theta$. Substituting this into \Cref{eq:FinalSolution} yeilds the reflected field:

\begin{align}\label{eq:ReflectionFromCylinder}
    \mathbf{E} (\mathbf{x}) = \int_{-\pi/2}^{\pi/2} e^{i \left|\mathbf{k}\right| \left(r_0 - 2\gamma \cos{\theta} \right)} \cos{\theta} U\left( \mathbf{x}\cdot\left( \cos{\theta} \mathbf{e}_1 + \sin{\theta}\mathbf{e}_2 \right) - \gamma \right) \mathbf{R}_3(2\theta) \mathbf{e}_2 e^{i \mathbf{R}_3(2\theta)\mathbf{k}\cdot\mathbf{x}} \mathrm{d}\theta
\end{align}

\begin{figure}[htbp]
    \centering
    \includegraphics[width=0.4\linewidth]{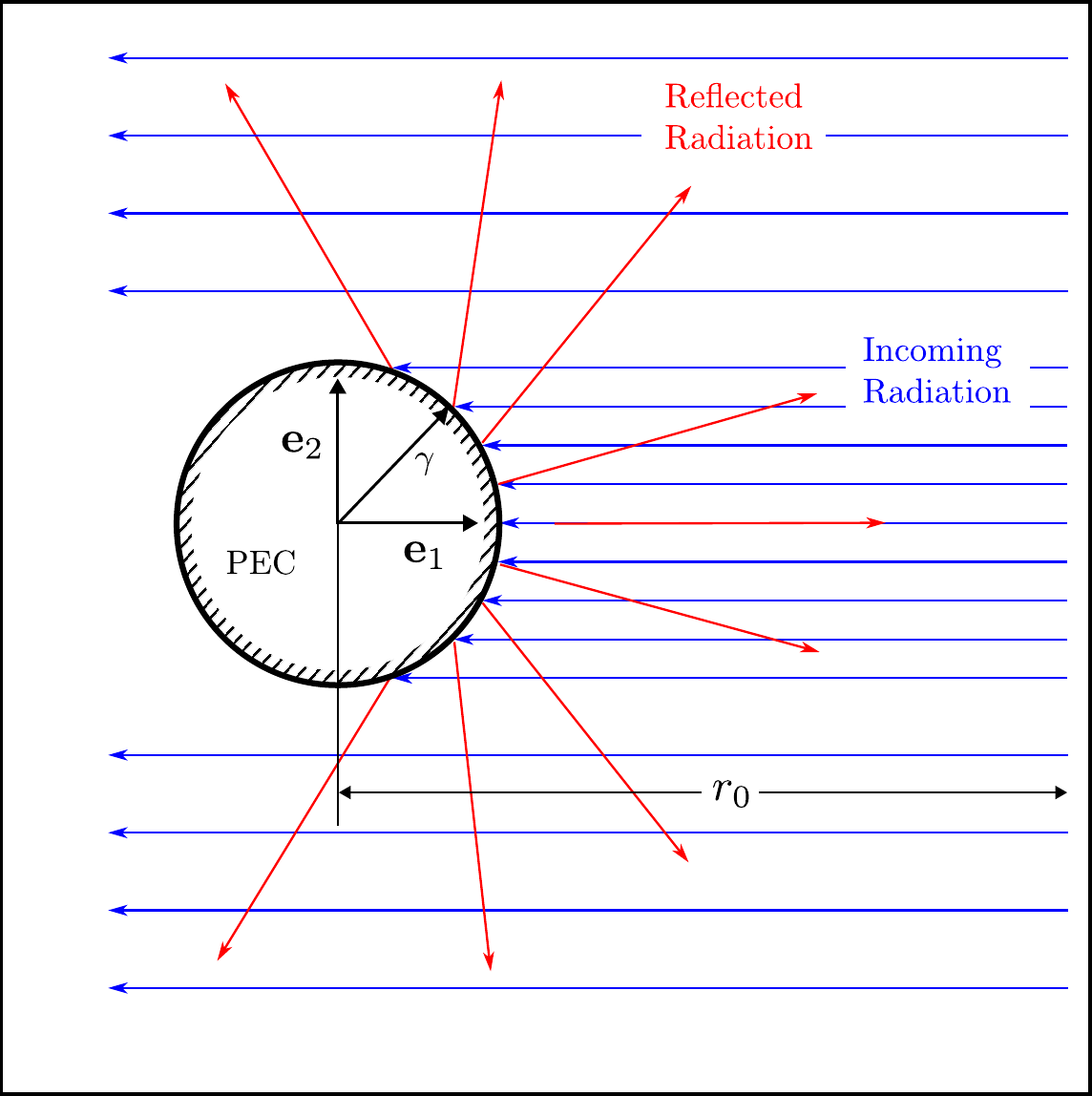}
    \includegraphics[width=0.4\linewidth]{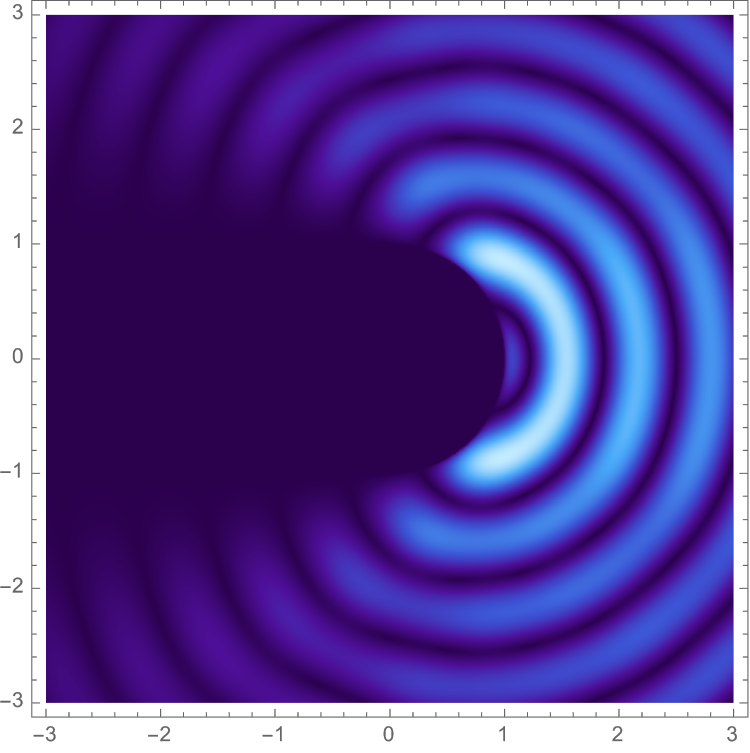}
    \caption{The schematic (left) shows the setup: incoming wave (blue), reflected rays (red), cylinder radius $\gamma$, and reference distance $r_0$ used to fix phase. The intensity plot (right) displays the resulting interference pattern (white: high intensity, dark blue: low intensity), with pronounced constructive and destructive interference when the wavelength is comparable to the cylinder radius.}
    \label{fig:PerfectCylinder}
\end{figure}

Each term in the integrand of \Cref{eq:ReflectionFromCylinder} corresponds to a distinct physical aspect of the reflection from the cylindrical surface:
\begin{enumerate}
    \item \textbf{Phase accumulation:} The factor $e^{i \left|\mathbf{k}\right| \left(r_0 - 2\gamma \cos{\theta} \right)}$ accounts for the extra optical path length between the reference plane and the point of reflection, including the $\pi$-phase jump absorbed into the amplitude.
    \item \textbf{Local incident amplitude:} The factor $\cos(\theta)$ describes how much of the incoming plane wave's flux intersects the surface at angle $\theta$. Points where the surface normal is nearly parallel to $\mathbf{e}_1$ receives the most flux; points near the top and bottom of the cylinder receives the least.
    \item \textbf{Boundary condition enforcement:} The unit step function $U(\mathbf{x} \cdot (\cos(\theta)\mathbf{e}_1 + \sin(\theta)\mathbf{e}_2) - \gamma)$ ensures that the reflected field is zero inside the cylinder, consistent with the physical boundary condition of a perfect conductor.
    \item \textbf{Reflected wavevectors:} The factor $e^{i(\mathbf{R}_3(2\theta)\mathbf{k})\cdot\mathbf{x}}$ represents the outgoing plane wave obtained by reflecting the incident wave across the tangent plane. Because the tangent plane has normal direction $\cos(\theta)\mathbf{e}_1 + \sin(\theta)\mathbf{e}_2$, the angle of incidence equals the angle of reflection.
    \item \textbf{Polarization rotation:} The term $\mathbf{R}_3(2\theta) \mathbf{e}_2$ represents the rotation of the polarization vector under reflection. For a perfectly conducting surface, the tangential component of $\mathbf{E}$ reverses the sign. Although often neglected in geometric optics, the rotation $2\theta$ captures this change in handedness, consistent with \cite{Swindell_71}.
\end{enumerate}

\section{Approximation by trapezoidal formula}\label{sec:ApproximationOfGeneralization}

In practice, evaluating the integral in \eqref{eq:FinalSolution} directly appears to be difficult. However, under mild smoothness and periodicity conditions on the integrand, the trapezoidal rule offers a remarkably efficient and exponentially convergent approximation \cite{trefethen_exponentially_2014}. This enables practical computation and opens the door to realizing such fields physically using a finite number of plane-wave sources.

The integral representation of twisted X-rays exhibit two key properties: integrands are analytic bounded functions, and they are periodic over the domain of integration. These features make the trapezoidal rule particularly well-suited, often outperforming more sophisticated quadrature schemes in terms of simplicity and convergence rate. Moreover, each discrete term in the approximation corresponds to a classical plane-wave, making the method physically interpretable in addition to being computationally attractive.

\begin{figure}[htbp]
  \centering
  \begin{subfigure}[b]{0.7\textwidth}
    \includegraphics[width=\textwidth]{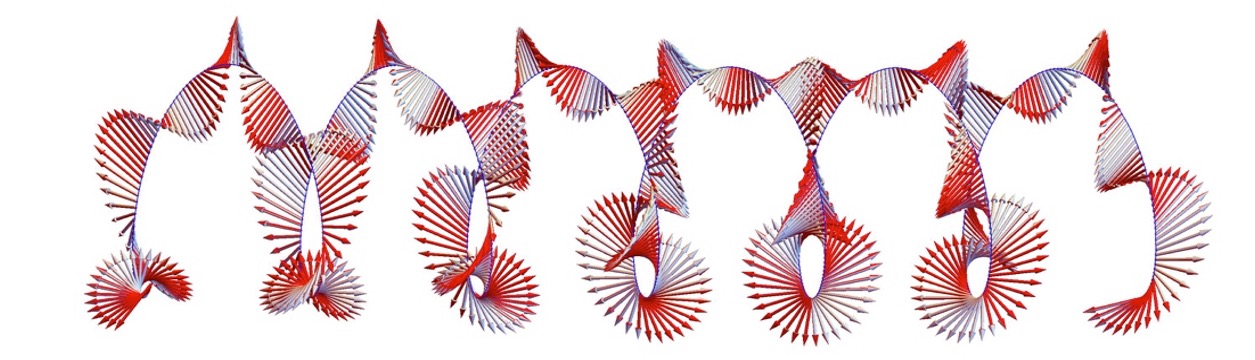}
    \caption{10 terms with maximum normalized error = $0.01$}
    \label{fig:convergence1}
  \end{subfigure}
  \begin{subfigure}[b]{0.7\textwidth}
    \includegraphics[width=\textwidth]{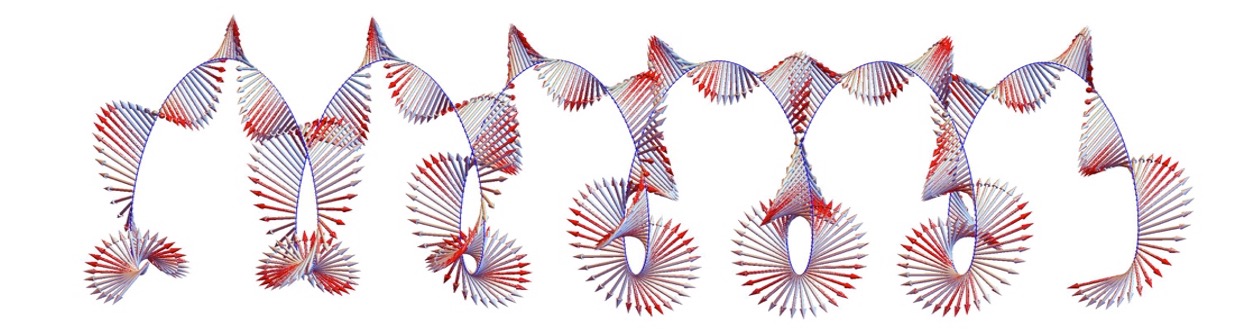}
    \caption{15 terms with maximum normalized error = $1.1 \times 10^{-6}$)}
    \label{fig:convergence2}
  \end{subfigure}
  \begin{subfigure}[b]{0.7\textwidth}
    \includegraphics[width=\textwidth]{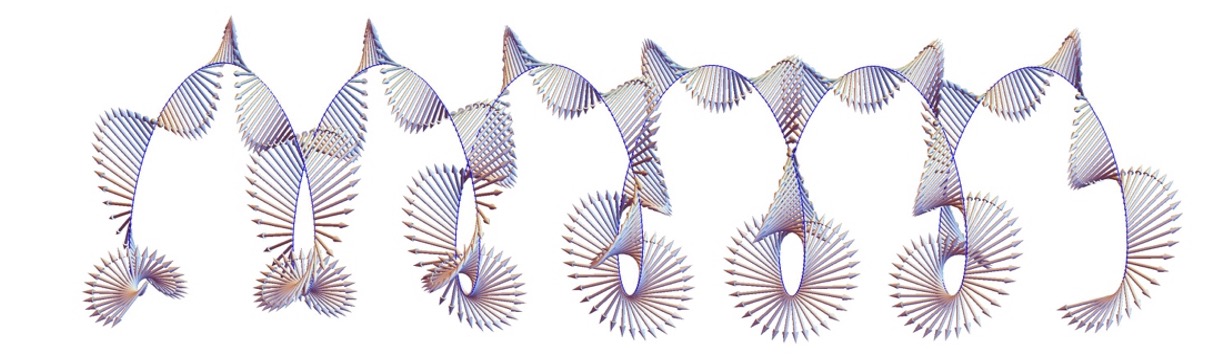}
    \caption{20 terms with maximum normalized error = $1.5 \times 10^{-11}$}
    \label{fig:convergence3}
  \end{subfigure}
  \caption{Each subfigure shows the approximation obtained using $N$ terms in the trapezoidal rule. The red arrows denote the exact electric field computed from the integral representation, and the gray arrows denote the numerical approximation. As $N$ increases from $10$ to $20$, the error decreases from $10^{-2}$ to $10^{-11}$, illustrating the exponential convergence characteristics of the trapezoidal rule for analytic and periodic integrands.}\label{fig:TwistedXRayConvergence}
\end{figure}

The \Cref{fig:TwistedXRayConvergence} demonstrates the rapid convergence of the trapezoidal rule for a representative twisted X-ray electric field along its Poynting vector. In these visualizations, the gray vectors represent the exact solution computed via integral representations, while the red vectors represnet the approximations using a finite number of terms.

The trapezoidal rule exhibits a similar convergence rate for the general representation \Cref{eq:FinalSolution}, provided that the integrand satisfies the necessary conditions for exponential convergence. The necessary conditions are:
\begin{enumerate}
    \item The integrand must be periodic and analytic in the integration variables, necessary to have uniformly and absolutely convergent Fourier series.
    \item The analytic continuation of integrand remain bounded in a strip of complex domain $\Im(z)>-a$ for some $a>0$.
\end{enumerate}

Assume now that the distributions $A(\theta_2,\theta_3)$ and $\theta_1(\theta_2,\theta_3)$ appearing in \Cref{eq:FinalSolution} are \emph{regular}, so that they may be treated as smooth functions on the domain $\theta_3\in [-\pi,\pi), \quad \theta_2\in[-\pi/2,\pi/2].$ Because $\theta_2$ and $\theta_3$ play the roles of polar and azimuthal angles on the sphere $\mathbb{S}^2$, these functions may naturally be viewed as functions defined on $\mathbb{S}^2$. Consequently, they inherit periodicity in $\theta_3$ (the longitudinal angle), but not in $\theta_2$ (the latitudinal angle). In particular, $F\left(\theta_2=-\frac{\pi}{2}, \theta_3 \right) \neq F\left(\theta_2=\frac{\pi}{2}, \theta_3 \right),$ and the same mismatch may appear in derivatives of any order. Thus, the integrand in \Cref{eq:FinalSolution} is periodic in $\theta_3$ with period $2\pi$, but not periodic in $\theta_2$, and therefore the trapezoidal rule in $\theta_2$  cannot exhibit exponential convergence.

A naive remedy would be to extend the domain in $\theta_2$ from $[-\pi/2,\pi/2]$ to $[-\pi,\pi)$ by reflecting the functions $A(\theta_2,\theta_3)$ and $\theta_1(\theta_2,\theta_3)$ across $\theta_2 = \pm\pi/2$. However, this straightforward reflection causes destructive cancellation: the contribution from the reflected region cancels that from the original domain, yielding zero. A more appropriate remedy is to apply the double Fourier-series extension method described in \cite{Merilees01011973,Orszag1974FourierSO}. In this construction, the functions are extended to the periodic domain $(\theta_2,\theta_3)\in[-\pi,\pi)\times[-\pi,\pi)$ by enforcing the parity-dependent transformation $$A(\theta_2,\theta_3)=A(\pi-\theta_2,\theta_3+\pi), \quad \theta_1(\theta_2,\theta_3)=\theta_1(\pi-\theta_2,\theta_3+\pi),$$ on the extended domain. This transformation compensates for the coordinate singularity at the poles by linking reflection in $\theta_2$ with a $\pi$-shift in $\theta_3$. As a result, integration over the extended domain yields exactly twice the value of the integral over the original domain, instead of canceling to zero.

However, an arbitrary pair of functions $A:[-\pi/2,\pi/2]\times[-\pi,\pi) \to \mathbb{C}$ and $\theta_1:[-\pi/2,\pi/2]\times[-\pi,\pi) \to \mathbb{R}$, when extended by this transformation to $[-\pi,\pi)\times[-\pi,\pi)$, will not in general be analytic on extended domain. To ensure analyticity, required for exponential convergence of the trapezoidal rule, we may impose additional mild constraints on the original functions. We outline several strategies:. 

\begin{enumerate}
    \item \textbf{Analytic composition:} Suppose $\tilde{A}:\mathbb{R}^3 \to \mathbb{C}$ and $\tilde{\theta}_1:\mathbb{R}^3 \to \mathbb{R}$ are analytic in all arguments. Define $A(\theta_2,\theta_3) = \tilde{A}\left(\mathbf{x}(\theta_2,\theta_3)\right)$ and $\theta_1(\theta_2,\theta_3) = \tilde{\theta}_1\!\left(\mathbf{x}(\theta_2,\theta_3)\right)$, where $\mathbf{x}(\theta_2,\theta_3)$ is the usual analytical mapping from $[-\pi/2,\pi/2]\times[-\pi,\pi)$ to the unit sphere $\mathbb{S}^2$. Because a composition of analytic functions is analytic, the resulting functions $A$ and $\theta_1$ are analytic on the integration domain, and their parity-dependent extensions remain analytic on the periodic square.
    \item \textbf{Vanishing intensity at antipodal points:} If the physical model prescribes the intensity of plane waves as a function of direction, the previous strategy may not be applicable. In such cases, analyticity on the extended domain can still be ensured by requiring that both functions $A$ and $\theta_1$ vanish smoothly at a pair of antipodal points on $\mathbb{S}^2$. These two points may then be mapped to $\theta_2=\pm\pi/2$. Under this constraint, the parity-dependent extension becomes both periodic and analytic on the integration domain, allowing exponential convergence of the trapezoidal rule.
    \item \textbf{Enforcing compatibility conditions:} A more general variant of the previous strategy allows $A$ and $\theta_1$ to take non-zero values at $\theta_2=\pm\pi/2$ provided their derivatives satisfy compatibility conditions at the reflection planes. Specifically, analyticity on the extended domain is possible if, for all non-negative integers $n$, $$\frac{\partial^n}{\partial \theta_2^n}   A(\theta_2 = \pm \pi/2,\theta_3) = (-1)^n \frac{\partial^n}{\partial \theta_2^n} A(\theta_2 = \pm \pi/2,\theta_3+\pi),$$ $$\frac{\partial^n}{\partial \theta_2^n}   \theta_1(\theta_2 = \pm \pi/2,\theta_3) = (-1)^n \frac{\partial^n}{\partial \theta_2^n} \theta_1(\theta_2 = \pm \pi/2,\theta_3+\pi).$$ These compatibility conditions ensure that the extended functions join smoothly across the reflection boundaries, required for analyticity.
\end{enumerate}

\begin{figure}[htbp]
  \centering
  \begin{subfigure}[b]{0.45\textwidth}
    \includegraphics[width=\textwidth]{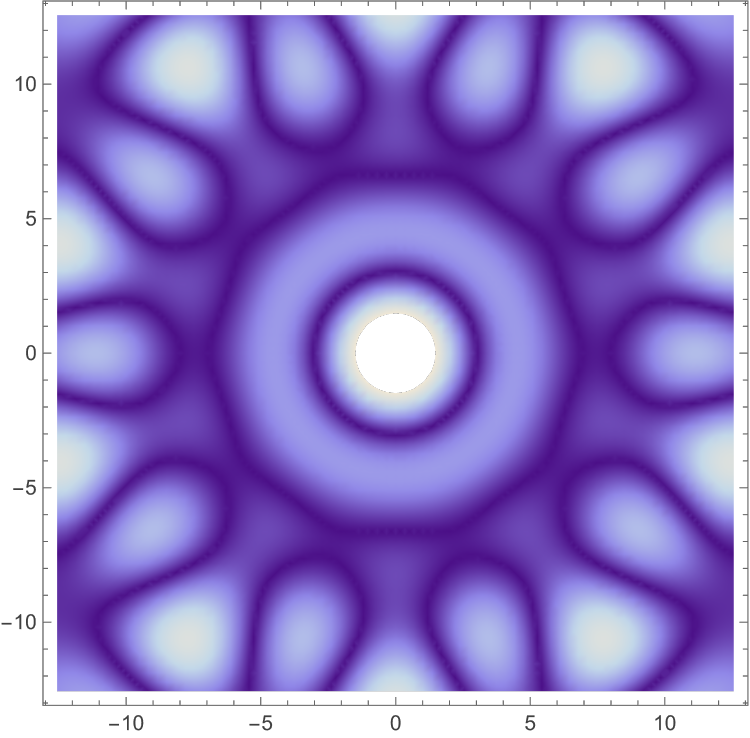}
    \caption{10 plane-waves }
    \label{fig:convergence1anew}
  \end{subfigure}
  \begin{subfigure}[b]{0.45\textwidth}
    \includegraphics[width=\textwidth]{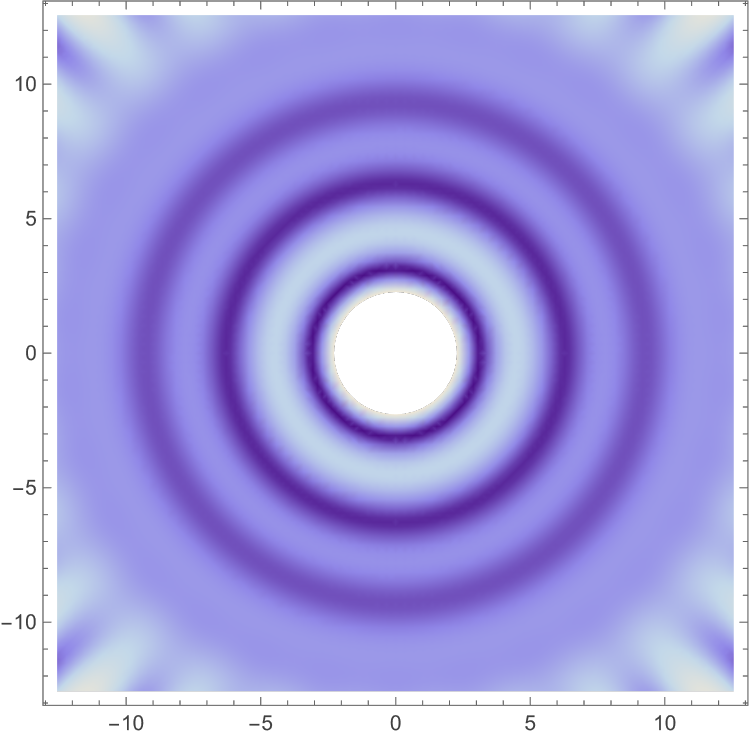}
    \caption{15 plane-waves }
    \label{fig:convergence2anew}
  \end{subfigure}\\
  \begin{subfigure}[b]{0.45\textwidth}
    \includegraphics[width=\textwidth]{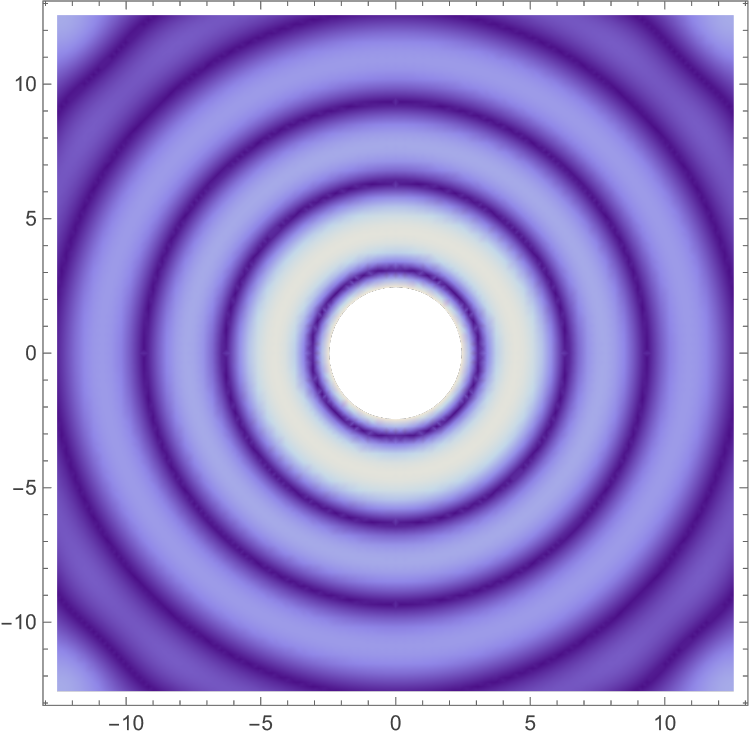}
    \caption{20 plane-waves }
    \label{fig:convergence3anew}
  \end{subfigure}
  \begin{subfigure}[b]{0.45\textwidth}
    \includegraphics[width=\textwidth]{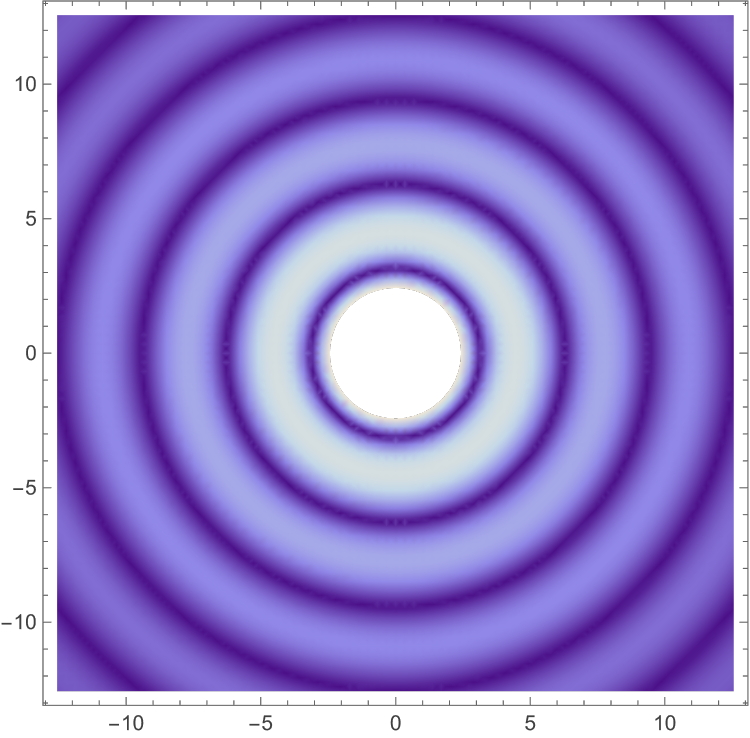}
    \caption{Reference solution}
    \label{fig:convergence4anew}
  \end{subfigure}\\
  \caption{Intensity plots showing trapezoidal rule convergence for a 3D Gaussian beam with nontrivial dependence on angular variables $\theta_2$ and $\theta_3$. White indicates high intensity and blue indicates low intensity. Subfigures (a)–(c) show approximations with increasing plane-wave terms; (d) shows the exact integral solution. The results demonstrate exponential convergence even in the non-degenerate case. Relative wall-clock times (computational costs) are $0.06$, $0.10$, and $0.25$ of the reference solution for $10$, $15$, and $20$ terms, respectively.} 
  \label{fig:GeneralisedTwistedXRayConvergence}
\end{figure}

A typical approximation of the generalized twisted X-ray field is given by:
\begin{align*}
    \mathbf{E}(\mathbf{y}) &= \int_{-\pi}^\pi \int_{-\pi/2}^{\pi/2} A(\theta_2,\theta_3) \mathbf{R}_3 (\theta_3) \mathbf{R}_2 (\theta_2) \mathbf{R}_1 \left(\theta_1(\theta_2,\theta_3)\right) \mathbf{n}_0 e^{i (\mathbf{R}_3 (\theta_3) \mathbf{R}_2 (\theta_2) \mathbf{k}_0 \cdot \mathbf{y})} \mathrm{d}\theta_2 \mathrm{d}\theta_3\\
    & \approx \frac{2\pi}{M}\sum_{\theta_{3_m}} \frac{\pi}{N}\sum_{\theta_{2_n}} A(\theta_{2_n},\theta_{3_m}) \mathbf{R}_3 (\theta_{3_m}) \mathbf{R}_2 (\theta_{2_n}) \mathbf{R}_1 \left(\theta_1(\theta_{2_n},\theta_{3_m})\right) \mathbf{n}_0 e^{i (\mathbf{R}_3 (\theta_{3_m}) \mathbf{R}_2 (\theta_{2_n}) \mathbf{k}_0 \cdot \mathbf{y})}.
\end{align*}
Here, $\theta_{3_m} = m\frac{2\pi}{M}$ and $\theta_{2_n} = n\frac{2\pi}{N}$ with $m = 0,1,\dots,M-1$ and $n = 0,1,\dots,N-1$. All previous examples in the paper reduce to a single angular integration. The Gaussian-beam example shown in \Cref{fig:GeneralisedTwistedXRayConvergence} is deliberately chosen to break this degeneracy. Here the amplitude function
\begin{align*}
    A(\theta_2,\theta_3) = \exp{-\frac{k^2}{w_o^2}\sin^2{\theta_2}}, \quad \theta_1(\theta_2,\theta_3)=0,
\end{align*} 
makes the integrand depend non-trivially on both variables, requiring two integrations in the representation formula. This illustrates that the full two-parameter integral is needed in general, and supports the theoretical result that two angular integrations suffice to represent all time-harmonic electromagnetic fields in $\mathcal{A}^3$. In this case, the integral involves two variables, requiring repeated application of the trapezoidal rule. The integral representation and its trapezoidal approximation offer several advantages:

\begin{itemize}
\item \textbf{Exactness:} The truncated approximation yields a field that \emph{exactly} satisfies the time-harmonic Maxwell equations in free space.
\item \textbf{Efficiency:} The trapezoidal rule provides a computationally inexpensive and exponentially convergent numerical approximation.
\item \textbf{Physical interpretability :} Each discrete term corresponds to a classical plane wave with a specific direction and phase, making the approximation directly implementable using standard plane-wave sources.
\end{itemize}

In the following subsections, we explore how these approximations can be applied to physically relevant problems, including radiation design and scattering analysis.

\subsection{Approximate solution of boundary value problems}
It has been shown that any bounded time-harmonic solution to Maxwell's equations in free space can be expressed using the integral formulation in \Cref{eq:FinalSolution}. This implies that any desired radiation field can, in principle, be designed using only plane-wave sources. However, in practice, it is not feasible to use an infinite number of plane waves. 
  The trapezoidal rule, discussed earlier, provides a practical way to overcome this limitation: by approximating the integral with a finite number of terms, we can construct radiation fields that closely match the desired behavior, while acknowledging that the discretization introduces some error. 

\begin{figure}[htbp]
    \centering
    \includegraphics[width=0.5\linewidth]{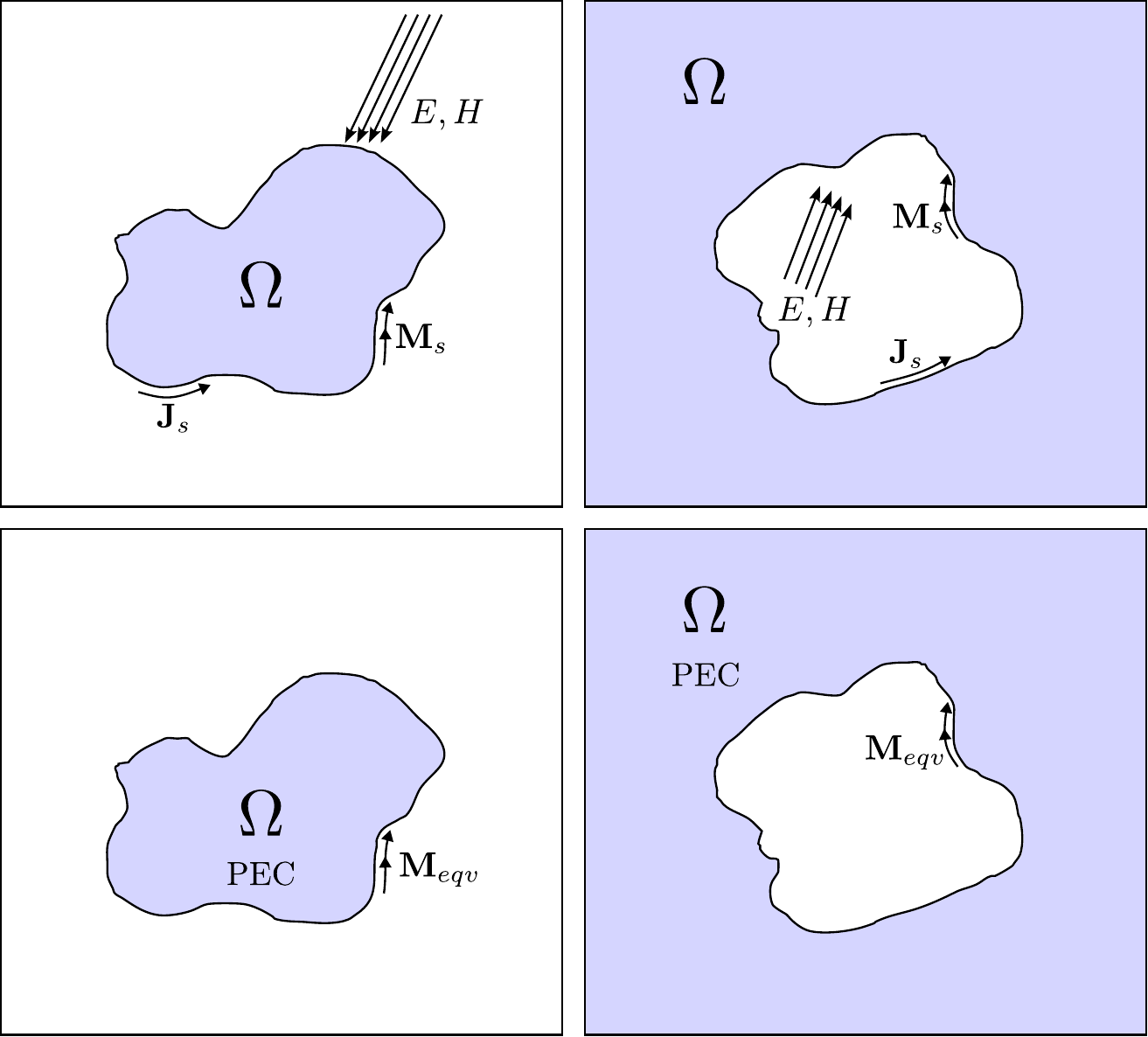}
    \caption{Schematic illustrating a typical interior (top-left) and exterior (top-right) boundary value problem. In both cases, the region of interest is denoted by $\Omega$. The bottom panels show the corresponding surface-equivalent formulations under a perfect electric conductor (PEC) boundary, following Schelkunoff's surface equivalence principle. The equivalent magnetic surface current $\mathbf{M}_{eqv}$ is placed on the boundary to reproduce the same fields in $\Omega$.  Prescribing $\mathbf{M}_{eqv}$ is equivalent to prescribing the tangential electric field on the boundary for both interior and exterior problems.}
    \label{fig:InteriorExteriorProblem}
\end{figure}

In general, a well-defined boundary value problem involves specifying one tangential field, either electric or magnetic, on the boundary. For an interior boundary value problem, this condition is sufficient. However, for an exterior boundary value problem, an additional condition on the field at infinity is required to ensure the uniqueness of the solution. Using the surface equivalence principle \cite{Stratton_Chu_1939,Schelkunoff1936}, the source terms and material properties in the exterior (or interior) region can be replaced by equivalent surface magnetic and electric currents. These equivalent currents generate the same fields in the region of interest. Thus, the boundary value problem can, in principle, be reformulated by prescribing an equivalent tangential electric field on the boundary (See \Cref{fig:InteriorExteriorProblem}).

We now show that it is possible to match the electric field \emph{exactly} at an arbitrarily large but finite number of points using a sufficient number of plane-waves. This is stronger than simply matching the boundary condition at a finite number of locations in two important ways. First, the matching points may lie anywhere in the domain, not only on the boundary. Second, all three components of the electric field can be matched at each point, not just the tangential components.

\begin{theorem}
    For any prescribed electric field values at finitely many points in a source-free homogeneous domain, there exists a finite superposition of plane waves whose electric field matches those prescribed values exactly, provided sufficiently many plane waves are included in the superposition.
\end{theorem}

\begin{proof}
Consider constructing the required electric field $\mathbf{E}(\mathbf{x})$ as a superposition of plane waves that matches a prescribed field values $\tilde{\mathbf{E}}(\mathbf{x}_i)$ at points $\mathbf{x}_i$,  $i=1,2,\dots,m$. We approximate the integral representation \Cref{eq:FinalSolution} using the trapezoidal rule to construct $\mathbf{E}(\mathbf{x})$:

\begin{align}
    \mathbf{E}(\mathbf{x}) &= \frac{2\pi}{M}\sum_{\theta_{3_m}} \frac{\pi}{N}\sum_{\theta_{2_n}} A(\theta_{2_n},\theta_{3_m}) \mathbf{R}_3 (\theta_{3_m}) \mathbf{R}_2 (\theta_{2_n}) \mathbf{R}_1 \left(\theta_1(\theta_{2_n},\theta_{3_m})\right) \mathbf{n}_0 e^{i (\mathbf{R}_3 (\theta_{3_m}) \mathbf{R}_2 (\theta_{2_n}) \mathbf{k}_0 \cdot \mathbf{x})}
\end{align}

For brevity, we condense the notation by defining a single index $j$ for each plane wave: $A_j = \frac{2\pi}{M} \frac{\pi}{N} A(\theta_{2_n},\theta_{3_m}), \quad \theta_{1_j}=\theta_1(\theta_{2_n},\theta_{3_m}), \quad \mathbf{R}_j=\mathbf{R}_3 (\theta_{3_m}) \mathbf{R}_2 (\theta_{2_n})$. Matching the field  values at each point $\mathbf{x}_i$ requires: 

\begin{align}
    \tilde{\mathbf{E}}(\mathbf{x}_i) &= \sum_{j=1}^n A_j \mathbf{R}_j \mathbf{R}_1(\theta_{1_j}) \mathbf{n}_0 e^{i \mathbf{R}_j \mathbf{k}_0 \cdot \mathbf{x}_i}, \qquad \forall \mathbf{x}_i \in \{ \mathbf{x}_1,\mathbf{x}_2,\dots,\mathbf{x}_m \}.
\end{align}

Here, $\mathbf{R}_j$ and $\mathbf{k}_0$ are fixed, and $\mathbf{n}_0$ is orthogonal to the axis of rotation $\mathbf{R}_1$. The free parameters are complex amplitudes $A_j$ and polarization angles $\theta_{1_j}$. Notably, $A_j \mathbf{R}_1(\theta_{1_j}) \mathbf{n}_0$ can be expressed as:
\begin{align}
    A_j \mathbf{R}_1(\theta_{1_j}) \mathbf{n}_0 = \alpha_j \mathbf{n}_0 + \beta_j \mathbf{n}_1,
\end{align}
where $\mathbf{n}_1 = \mathbf{R}_1(\pi/2) \mathbf{n}_0$ is orthogonal to both $\mathbf{n}_0$ and the rotation axis, and $\alpha_j, \beta_j \in \mathbb{C}$ are uniquely determined by $A_j$ and $\theta_{1_j}$. This reduces the matching condition to a linear algebra problem:
\begin{align}
    \mathcal{E}_i &= \sum_{j=1}^{2n} \mathcal{P}_{ij} \mathcal{A}_j, \qquad \forall i=1,2,\dots,3m.
\end{align}
Or in matrix form:
\begin{align}
\boldsymbol{\mathcal{E}} &=  \boldsymbol{\mathcal{P}} \boldsymbol{\mathcal{A}}.
\end{align}

Here,
\begin{itemize}
    \item $\boldsymbol{\mathcal{E}}\in\mathbb{C}^{3m}$ contains the target field values,
    \item $\mathcal{A}\in\mathbb{C}^{2n}$ contains the unknowns $\alpha_j$ and $\beta_j$,
    \item $\boldsymbol{\mathcal{P}}\in\mathbb{C}^{3m\times2n}$ contains the rotated plane-wave evaluated at the sampling points, with $\mathcal{P}_{ij} = \mathbf{R}_j \mathbf{R}_1(\theta_{1_j}) \mathbf{n}_0 e^{i \mathbf{R}_j \mathbf{k}_0 \cdot \mathbf{x}_i}$ and $\theta_{1_j}\in\{0,\pi/2\}$.
\end{itemize}

A solution exists whenever $2n > 3m$ \cite{strang_linear_2006} and $\mathcal{P} \mathcal{P}^H$ is full rank. $\mathcal{P} \mathcal{P}^H$ is generally full rank for distinct sampling points $\mathbf{x}_i$. The minimum norm solution is given by the Moore-Penrose pseudoinverse:
\begin{align}
\boldsymbol{\mathcal{A}} &=  \boldsymbol{\mathcal{P}}^H \left( \boldsymbol{\mathcal{P}} \boldsymbol{\mathcal{P}}^H \right)^{-1} \boldsymbol{\mathcal{E}},
\end{align}
where $\boldsymbol{\mathcal{P}}^H$ denotes the conjugate transpose of $\boldsymbol{\mathcal{P}}$. The pseudoinverse yields the smallest energy configuration of complex amplitudes needed to reproduce the desired field at the specified points. 
\end{proof}

\subsection{Interference patterns for the truncated icosahedron}

Another application of this work is the identification of  electromagnetic fields that produce discrete diffraction patterns when interacting with materials exhibiting specific symmetry. A standard plane-wave produces such discrete patterns when interacting with crystals. For materials with helical symmetry, a corresponding radiation design is provided in \cite{friesecke_twisted_2016}. While the full inverse problem remains challenging for other symmetry groups, one can nonetheless construct incoming radiation fields that generate constructive interference in targeted directions. This section focuses on the case of a truncated-icosahedral structure and demonstrates how generalized twisted X-rays can be tailored to maximize intensity along a prescribed direction.

Consider a structured material illuminated by an engineered time-harmonic field $\mathbf{E}_0(\mathbf{y}) e^{-i \omega t}$. Let $\rho(\mathbf{y})$ denote the electronic density of a material. In the far-field approximation \cite{friesecke_twisted_2016}, the outgoing electric and magnetic field can be given by:

\begin{align*}
    \mathbf{E}_{out}(\mathbf{x},t) &= -c_{el} \frac{e^{i(\mathbf{k}'(\mathbf{x})\cdot\mathbf{x}-\omega t)}}{|\mathbf{x}-\mathbf{x}_c|} \left( \mathbf{I} - \frac{\mathbf{k}'(\mathbf{x})}{|\mathbf{k}'(\mathbf{x})|} \otimes \frac{\mathbf{k}'(\mathbf{x})}{|\mathbf{k}'(\mathbf{x})|}\right) \int_{\mathbb{R}^3} \mathbf{E}_0 (\mathbf{y}) \rho (\mathbf{y}) e^{-i\mathbf{k}'(\mathbf{x})\cdot\mathbf{y}} \mathrm{d}\mathbf{y},\\
    \mathbf{B}_{out}(\mathbf{x},t) &= \frac{1}{\omega} \mathbf{k}' (\mathbf{x}) \times \mathbf{E}_{out}(\mathbf{x},t),
\end{align*}
where the outgoing wavevector is $\mathbf{k}'(\mathbf{x})=\frac{\omega}{c}\frac{\mathbf{x}-\mathbf{x}_c}{|\mathbf{x}-\mathbf{x}_c|}$, $c_{el}$ is constant, $\omega$ is the angular frequency of the incoming radiation, $\mathbf{x}_c$ is a reference point inside the sample. To maximize the intensity of the outgoing field in specific direction $\mathbf{k}'(\mathbf{x}_0)$, we evaluate the electric field at $\mathbf{x}_0$:
\begin{align*}
    \mathbf{E}_{out} (\mathbf{x}_0,t) = e^{-i \omega t} c(\mathbf{x}_0) \mathbf{P}_{\mathbf{k}'(\mathbf{x}_0)} \int_{\mathbb{R}^3} \mathbf{E}_0 (\mathbf{y}) \rho(\mathbf{y}) e^{-i\mathbf{k}'(\mathbf{x}_0)\cdot\mathbf{y}} \mathrm{d}\mathbf{y},
\end{align*}
where $c(\mathbf{x}_0) = -c_{el} \frac{e^{i\mathbf{k}'(\mathbf{x}_0)\cdot\mathbf{x}_0}}{|\mathbf{x}_0-\mathbf{x}_c|}$ and $\mathbf{P}_{\mathbf{k}'(\mathbf{x}_0)} = \left( \mathbf{I} - \frac{\mathbf{k}'(\mathbf{x}_0)}{|\mathbf{k}'(\mathbf{x}_0)|} \otimes \frac{\mathbf{k}'(\mathbf{x}_0)}{|\mathbf{k}'(\mathbf{x}_0)|}\right)$ is the projection tensor onto the plane orthogonal to the outgoing wavevector. Substituting the generalized twisted X-ray representation for $\mathbf{E}_0$ yields
\begin{align*}
    \mathbf{E}_{out} (\mathbf{x}_0) & = c(\mathbf{x}_0) \mathbf{P}_{\mathbf{k}'(\mathbf{x_0})} \int_{\mathbb{R}^3}\int_{-\pi}^\pi \int_{-\pi/2}^{\pi/2} A(\theta_2,\theta_3) \mathbf{R}_3 (\theta_3) \mathbf{R}_2 (\theta_2) \mathbf{R}_1 \left(\theta_1(\theta_2,\theta_3)\right) \mathbf{n}_0 \\&\qquad \rho (\mathbf{y})e^{i\left(\mathbf{R}_3 (\theta_3) \mathbf{R}_2 (\theta_2) \mathbf{k}_0 - \mathbf{k}'(\mathbf{x}_0) \right)\cdot\mathbf{y}} \mathrm{d}\theta_2 \mathrm{d}\theta_3 \mathrm{d}\mathbf{y}. \\
\end{align*}

When only a finite number of plane wave sources are used, the outgoing field becomes the discrete sum:
\begin{align*}
    \mathbf{E}_{out} (\mathbf{x}_0) & =c(\mathbf{x}_0)  \frac{2\pi}{M}\sum_{\theta_{3_m}} \frac{\pi}{N}\sum_{\theta_{2_n}} A(\theta_{2_n},\theta_{3_m}) \mathbf{P}_{\mathbf{k}'(\mathbf{x_0})}\mathbf{R}_3 (\theta_{3_m}) \mathbf{R}_2 (\theta_{2_n}) \mathbf{R}_1 \left(\theta_1(\theta_{2_n},\theta_{3_m})\right) \mathbf{n}_0 \\&\qquad \mathcal{F}(\rho )\left(\mathbf{R}_3 (\theta_{3_m}) \mathbf{R}_2 (\theta_{2_n}) \mathbf{k}_0 - \mathbf{k}'(\mathbf{x}_0) \right), \\
\end{align*}
where $\mathcal{F}(\rho) (\mathbf{k})$ denotes Fourier transform of the electronic density. We model $\rho$ as a sum of Dirac delta distributions located at a vertices of a truncated icosahedron (a buckyball-type structures). These vertices are even permutations of
\begin{align*}
    (0,\pm 1,\pm 3\varphi), \quad (\pm 1 ,\pm(2+ \varphi), \pm 2 \varphi), \quad (\pm 2, \pm (1+2\varphi) ,\pm \varphi),
\end{align*}
where $\varphi = (1+\sqrt{5})/2$ is the golden ratio. The objective is to choose complex amplitudes $A(\theta_{2_n},\theta_{3_m})$ to maximize the intensity, proportional to $|\mathbf{E}_{out}(\mathbf{x}_0)|$, subject to a physical constraint $|A(\theta_{2_n},\theta_{3_m})| \leq1$, which corresponds to a bound on the individual plane-wave energies. To maximize intensity along a given  direction, say $\mathbf{k}'(\mathbf{x}_0)=|\mathbf{k}'(\mathbf{x}_0)|\mathbf{e}_2$, we apply the method developed in \Cref{sec:MaxMag}. In this configuration:
\begin{itemize}
    \item the resulting vectors lie in the $\mathbf{e}_1 - \mathbf{e}_3$ plane due to the projection tensor, and
    \item they are real-valued because for every vertex $\mathbf{x}$, there exists a corresponding vertex at $-\mathbf{x}$, ensuring $\mathcal{F}(\rho) (\mathbf{k})\in \mathbb{R}$. 
\end{itemize}

The method in \Cref{sec:MaxMag} provides only an implicit optimality condition, and therefore we require additional information about the direction of the resultant outgoing electric field $\mathbf{E}_{out} (\mathbf{x}_0)$ in order to determine the maximum achievable intensity and the corresponding amplitudes $A(\theta_{2_n},\theta_{3_m})$. This direction can be identified by evaluating the intensity as a function of the orientation of $\mathbf{E}_{out} (\mathbf{x}_0)$ within the $\mathbf{e}_1$-$\mathbf{e_3}$ plane, which is the range of the projection tensor. For the truncated-icosahedron configuration defined by the coordinates above, we find that the maximum intensity along the $\mathbf{e}_2$ axis is attained when the outgoing field $\mathbf{E}_{out} (\mathbf{x}_0)$ aligns with the $\mathbf{e}_1$ direction. This rule is independent of the discretization parameters $N$ and $M$, enabling us to propose the corresponding incoming radiation in the continuous limit:

\begin{align*}
    \mathbf{E}_0(\mathbf{y}) &= \int_{-\pi}^\pi \int_{-\pi/2}^{\pi/2} A(\theta_2,\theta_3) \mathbf{R}_3 (\theta_3) \mathbf{R}_2 (\theta_2) \mathbf{R}_1 \left(\theta_1(\theta_2,\theta_3)\right) \mathbf{n}_0 e^{i (\mathbf{R}_3 (\theta_3) \mathbf{R}_2 (\theta_2) \mathbf{k}_0 \cdot \mathbf{y})} \mathrm{d}\theta_2 \mathrm{d}\theta_3,
\end{align*}
with,
\begin{align*}
    A\left( \theta_2,\theta_3\right) = \operatorname{sign}\left(\mathbf{e}_1 \cdot \mathbf{P}_{\mathbf{k}'(\mathbf{x_0})}\mathbf{R}_3 (\theta_3) \mathbf{R}_2 (\theta_2) \mathbf{R}_1 \left(\theta_1(\theta_2,\theta_3)\right) \mathbf{n}_0  \mathcal{F}(\rho )\left(\mathbf{R}_3 (\theta_3) \mathbf{R}_2 (\theta_2) \mathbf{k}_0 - \mathbf{k}'(\mathbf{x}_0) \right)  \right),
\end{align*}
where $\mathbf{x}_0$ lies along $\mathbf{e}_2$. Although the exact position of $\mathbf{x}_0$ affects only the overall scaling factor  $c(\mathbf{x}_0)$, the angular distribution remains unchanged. The resulting intensity pattern is shown in \Cref{fig:IntensityOfTruncatedIcosahedron} for a particular frequency $\omega$.

\begin{figure}[htbp]
  \centering
  \begin{subfigure}[b]{0.48\textwidth}
    \includegraphics[width=\textwidth]{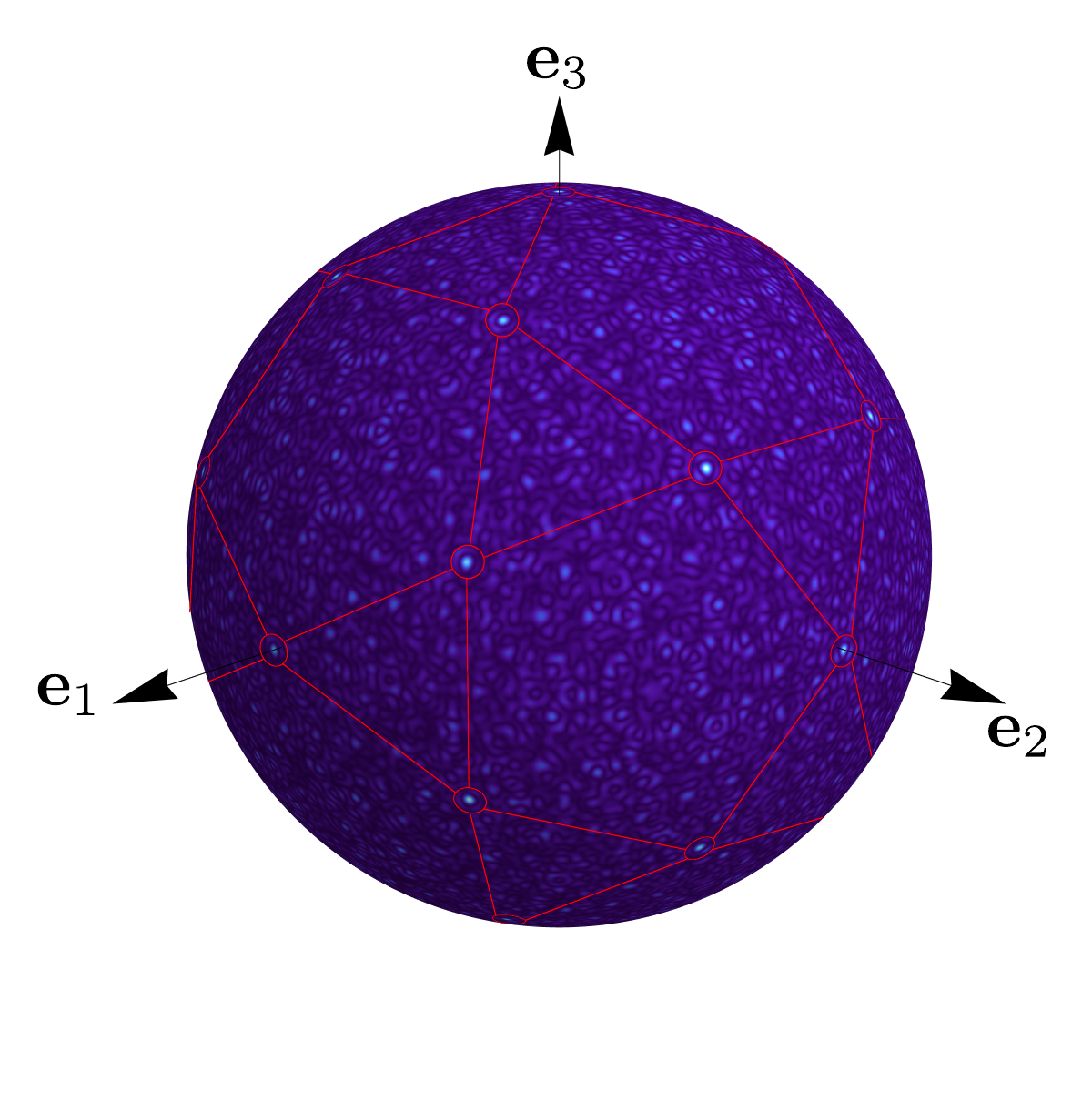}
    \caption{Truncated icosahedral structure}
    \label{fig:IntensityOfTruncatedIcosahedron}
  \end{subfigure}
  \hfill
  \begin{subfigure}[b]{0.48\textwidth}
    \includegraphics[width=\textwidth]{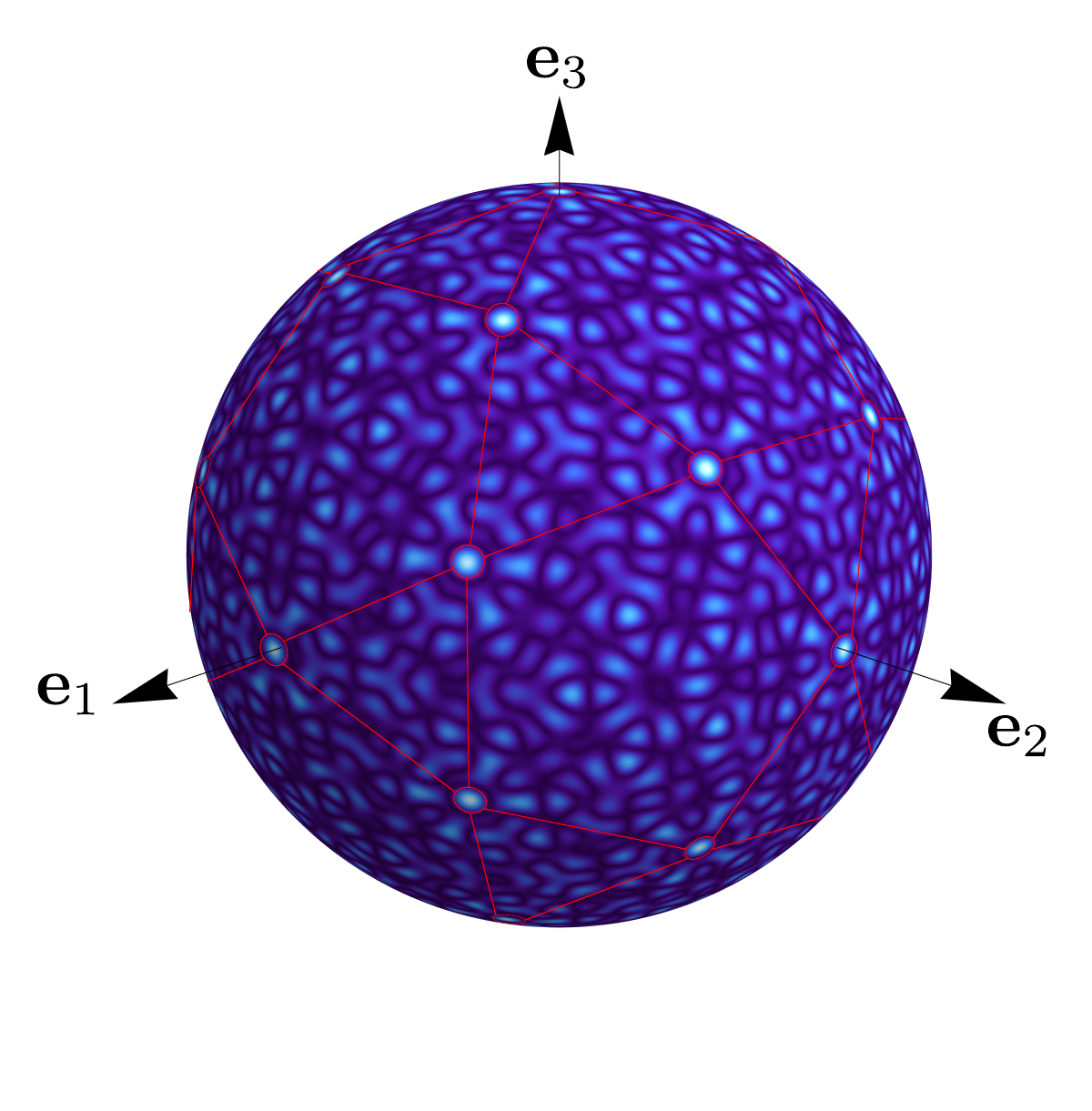}
    \caption{Icosahedral structure}
    \label{fig:IntensityOfIcosahedron}
  \end{subfigure}
  \caption{Simulated outgoing intensity (with red lines overlaid) resulting from interaction with (a) a truncated icosahedron and (b) an icosahedron. The incoming radiation is optimized to maximize intensity along the $\mathbf{e}_2$ axis. White regions correspond to high intensity and blue regions to low intensity. The strong bright spots align approximately with the vertices of an icosidodecahedron, whose edges are indicated by the red lines. The strong bright spots reflect the underlying symmetry of the structure.}
  \label{fig:CombinedIntensityPlots}
\end{figure}

The white regions corresponds to high intensity, while the blue regions corresponds to the low intensity. Although the incoming radiation is optimized to maximize intensity only along $\mathbf{e}_2$ axis, the underlying symmetry of the structure and its Fourier transform leads to additional bright regions along other directions. These bright regions cluster around the vertices of an icosidodecahedron, given (up to scaling) by the even permutations of
$$\left(0,0,\pm \varphi\right), \quad \left(\pm \frac{\varphi^2}{2},\pm \frac{\varphi}{2},\pm \frac{1}{2} \right).$$

To further illustrate the affect of symmetry, we repeat the analysis using an icosahedral structure, with electronic density modeled by Dirac delta distributions at the even permutations of $\left(0,\pm 1, \pm \varphi \right)$). The resulting intensity pattern is shown in \Cref{fig:IntensityOfIcosahedron}. These intensity patterns may reveal structural information if the corresponding inverse problem is solved: for example, determining whether the material exhibits icosahedral, truncated-icosahedral or icosidodecahedral geometry, identifying the orientation of the structure, and estimating the atomic bond lengths from the observed intensity distribution. We emphasize, however, that the inverse problem itself is not  solved in this work.

\section{Other relevant problems}\label{sec:OtherRelevantProblem}
The representation formulas developed in this paper are derived for the electric and magnetic fields governed by Maxwell's equations. However, the governing differential equations—specifically the Helmholtz equations—also appear in a wide range of other physical problems. Many wave phenomena, including those in acoustic and elasticity, are governed by equations that reduce to scalar or vector Helmholtz equations when seeking time-harmonic solutions.

This connection allows us to extend the use of our representation formulas beyond electromagnetism. In the following subsections, we present two specific examples that illustrate how these formulas can be adapted to solve problems in other domains. These examples include spherical waves, commonly used to model radiation from a point source in acoustic, and elastic waves, which describe deformation in solid media.

\subsection{Spherical waves}
The pressure waves, such as sound waves, in a homogeneous, isotropic medium are governed by the scalar wave equation, which can be derived from the Euler's equations for fluid dynamics. The scalar wave equation reduces to the scalar Helmholtz equation when seeking time-harmonic solutions: 

\begin{align}
    \nabla\cdot\nabla p_0(\mathbf{x}) = -(\omega/v)^2 p_0(\mathbf{x}),
\end{align}
where $p(\mathbf{x},t) = e^{-i \omega t} p_0(\mathbf{x})$ is the pressure field and $v$ is the acoustic velocity in the medium. Using the approach developed in the paper, the solution to the scalar Helmholtz equation can be expressed as a superposition of plane waves:

\begin{align}
    p_0(\mathbf{x}) = \int_{-\pi}^{\pi} \int_{-\pi/2}^{\pi/2} a(\theta_3,\theta_2) e^{i \mathbf{R}_3(\theta_3) \mathbf{R}_2(\theta_2) \mathbf{k}.\mathbf{x}} \mathrm{d}\theta_2 \mathrm{d}\theta_3,
\end{align}
where $a(\theta_3,\theta_2)$ represents the amplitude of the pressure wave travelling in the direction defined by the polar and azimuthal angles $\theta_2$ and $\theta_3$, respectively. The rotation tensors $\mathbf{R}_2$ and $\mathbf{R}_3$ are defined about orthogonal axes $\mathbf{e}_2$ and $\mathbf{e}_3$, respectively, and $\mathbf{k}= (\omega/v) \mathbf{e}_2 \times \mathbf{e}_3$. Such solutions were previously discussed in \cite{Whittaker1903}, which also addressed vector Helmholtz equations. However, in electromagnetism, the additional divergence-free condition introduces constraints that we
have overcome in \Cref{sec:GeneralizationOfTwistedXrays}.

A particularly important solution to the scalar Helmholtz equation is the spherical wave, also known as the fundamental solution: 

\begin{align}\label{eq:FundamentalSolution}
    \Phi(r) = \frac{1}{4\pi} \frac{e^{i(\omega/v) r}}{r}
\end{align}
This solution is widely used in Green's function representations of the scalar Helmholtz equation and is often referred to as a Helmholtz representation, or Green's formula \cite{colton2013inverse}. In \cite{devaney_multipole_1974}, it is shown that these solutions can be expressed as a superposition of plane waves, provided the wave vectors are allowed to be complex, a topic discussed in \Cref{sec:ComplexWaveVector}. Such solutions can be represented by \Cref{eq:HelmholtzScalarSolutionDissipative} by choosing $\hat{u}$ and $\theta_1$ according to the sign of $(\hat{\mathbf{k}}_1\times\hat{\mathbf{k}}_2) \cdot \mathbf{x}$. For $(\hat{\mathbf{k}}_1\times\hat{\mathbf{k}}_2) \cdot \mathbf{x}>0$, we set
\[
\hat{u}(\theta_2,\theta_3,\phi)=
\begin{cases}
\cos{\theta_2}, & \phi=0,\; -\frac{\pi}{2} \leq \theta_2<0,\\
-\sec^2{\phi}, & 0<\phi\leq \frac{\pi}{2}, \; \theta_2=0,\\
0, & \text{otherwise,}
\end{cases}
\qquad \text{and} \qquad \theta_1(\theta_2,\theta_3,\phi)=\frac{\pi}{2}.
\]
For $(\hat{\mathbf{k}}_1\times\hat{\mathbf{k}}_2) \cdot \mathbf{x}<0$, we set
\[
\hat{u}(\theta_2,\theta_3,\phi)=
\begin{cases}
\cos{\theta_2}, & \phi=0,\;  0 <\theta_2 \leq\frac{\pi}{2},\\
\sec^2{\phi}, & 0<\phi\leq \frac{\pi}{2}, \; \theta_2=0,\\
0, & \text{otherwise,}
\end{cases}
\qquad \text{and} \qquad \theta_1(\theta_2,\theta_3,\phi)=-\frac{\pi}{2}.
\]
The use of complex wave vectors in Helmholtz representations has, in fact, been well established for many decades, with extensive early developments appearing in classical scattering and radiation theory, \cite{cappellin_properties_2008,Booker_1950}.

Interestingly, although spherical waves are fundamental in scalar wave theory, they cannot exist in electromagnetism or in general transverse, time-harmonic vector fields without breaking spherical symmetry. This limitation arises from the Hairy Ball Theorem, which prohibits the existence of a continuous, non-vanishing tangent vector field on a sphere. Nevertheless, the fundamental solution of the scalar Helmholtz equation appears in many electromagnetic field expressions. The field produced by a typical antenna—modeled by a Hertzian dipole—or by the electrons in a material excited by an external plane-wave source \cite{friesecke_twisted_2016, jaykar_macroscopic_2026}, naturally exhibits such fundamental solutions. A typical expression for the field generated by a Hertzian dipole at the origin—or equivalently for the field of an electron at the origin excited by a plane wave—is given by:
\begin{align*}
    \mathbf{E} (\mathbf{x},t) &= \frac{k}{|\mathbf{x}|} \left( \mathbf{I}-\frac{\mathbf{x}}{|\mathbf{x}|}\otimes\frac{\mathbf{x}}{|\mathbf{x}|} \right)\mathbf{F} e^{-i\omega (t-|\mathbf{x}|/c)}\ .
\end{align*}
Here, $k$ is a constant, $\mathbf{F}$ is the direction along which the dipole or electron oscillates. In the case of antennas, this represents the orientation of an omnidirectional antenna. Note that such fields can be written as $\mathbf{E} (\mathbf{x},t) = 4\pi k e^{-i\omega t} \left( \mathbf{I}-\frac{\mathbf{x}}{|\mathbf{x}|}\otimes\frac{\mathbf{x}}{|\mathbf{x}|} \right)\mathbf{F}  \Phi(|\mathbf{x}|)$, where $\Phi(r) = \frac{1}{4\pi} \frac{e^{i(\omega/c) r}}{r}$. 

\subsection{Elastic waves}
We have given examples of physical systems governed by both the vector Helmholtz equation (electromagnetic waves) and the scalar Helmholtz equation (acoustic waves). In this section, we consider elastic waves in isotropic solids, which require \emph{both} scalar and vector Helmholtz equations to model their two fundamental modes of propagation \cite{Meyers_1994}.

Longitudinal (or dilatational) waves are governed by the wave equation for the divergence of the displacement field, which reduces to a scalar Helmholtz equations under the time-harmonic assumption. Shear (or rotational) waves, on the other hand, are governed by the wave equation for the curl of the displacement field, which similarly reduces to a vector Helmholtz equation in the time-harmonic case.

Consider a linear isotropic elastic material with constitutive relation:

\begin{align}
    \boldsymbol{\sigma} = \lambda \mathrm{tr} \left(\frac{\nabla\mathbf{u} + \nabla\mathbf{u}^T}{2} \right) \mathbf{I} + 2\mu \left(\frac{\nabla\mathbf{u} + \nabla\mathbf{u}^T}{2} \right),
\end{align}
where $\boldsymbol{\sigma}$ is the Cauchy stress tensor, $\mathbf{u}$ is the displacement field, and $\lambda, \mu$ are the Lamé parameters. The balance of linear momentum $\nabla\cdot\boldsymbol{\sigma} = \rho \Ddot{\mathbf{u}}$, where $\rho$ is the mass density, leads to:

\begin{align}\label{eq:LinearMomentumBalance}
    (\lambda + \mu) \nabla \left( \nabla\cdot\mathbf{u} \right) + \mu \nabla\cdot\nabla\mathbf{u} = \rho \Ddot{\mathbf{u}}
\end{align} 

Taking the curl of \Cref{eq:LinearMomentumBalance} yields the wave equation for the rotation field $\nabla\times\mathbf{u}$: 

\begin{align}
    \mu \nabla\cdot\nabla \left(\nabla\times\mathbf{u}\right) = \rho \Ddot{\left( \nabla\times\mathbf{u}\right)}
\end{align}
This equation governs shear waves, which propagate at speed $v_s = \sqrt{\frac{\mu}{\rho}}$. While the waves produced in air by a tuning fork are longitudinal, the elastic waves within the tuning fork are a classic example of the shear waves.

Assuming a time-harmonic solution with angular frequency $\omega_s$, the rotation field can be represented using the framework developed above:

\begin{align}
    \nabla\times\mathbf{u} (\mathbf{x},t) = e^{-i\omega_s t} \int_{-\pi}^{\pi} \int_{-\pi/2}^{\pi/2} a_s(\theta_2,\theta_3) \mathbf{R}_3(\theta_3) \mathbf{R}_2(\theta_2) \mathbf{R}_1(\theta_1(\theta_2,\theta_3)) \boldsymbol{\omega}_s e^{i \mathbf{R}_3(\theta_3)\mathbf{R}_2(\theta_2)\mathbf{k}_s\cdot\mathbf{x}} \mathrm{d}\theta_2 \mathrm{d}\theta_3,
\end{align}
where:
\begin{itemize}
    \item $\omega_s$ is the angular frequency,
    \item $a_s (\theta_2,\theta_3)$ is the amplitude of a shear wave traveling in direction $\mathbf{R}_3(\theta_3)\mathbf{R}_2(\theta_2)\mathbf{k}_s$,
    \item $\mathbf{e}_1$,$\mathbf{e}_2$,$\mathbf{e}_3$ are the orthonormal axes associated with the rotations $\mathbf{R}_1$,$\mathbf{R}_2$,$\mathbf{R}_3$,
    \item $\mathbf{k}_s$ is a wave vector with $|\mathbf{k}_s| = \omega_s/v_s$ and aligned with the axis of $\mathbf{R}_1$,
    \item $\boldsymbol{\omega}_s$ is a rotation vector aligned with the axis of $\mathbf{R}_3$,
    \item $\theta_1$ specifies the orientation of $\mathbf{\omega}_s$ in the plane perpendicular to the propagation direction.
\end{itemize}  

The rotation vector $\boldsymbol{\omega}_s$ must be orthogonal to the wave vector $\mathbf{k_s}$, consistent with divergence-free nature of the curl field. Choosing $\mathbf{k}_s \perp\mathbf{\omega}_s$ aligns naturally with the integral representation developed earlier.

Taking the divergence of \Cref{eq:LinearMomentumBalance} yields the wave equation for the dilatation field:

\begin{align}
    (\lambda + 2\mu) \nabla\cdot\nabla \left(\nabla\cdot\mathbf{u}\right) = \rho \Ddot{\left( \nabla\cdot\mathbf{u}\right)}
\end{align}

This equation governs longitudinal waves, which propagate at speed $v_l = \sqrt{\frac{\lambda + 2\mu}{\rho}}$. 

Under a time-harmonic assumption with frequency $\omega_l$, the dilatation can be represented as:

\begin{align}
    \nabla\cdot\mathbf{u} (\mathbf{x},t) = e^{-i\omega_l t} \int_{-\pi}^{\pi} \int_{-\pi/2}^{\pi/2} a_l(\theta_2,\theta_3)  e^{i \mathbf{R}_3(\theta_3)\mathbf{R}_2(\theta_2)\mathbf{k}_l\cdot\mathbf{x}} \mathrm{d}\theta_2 \mathrm{d}\theta_3,
\end{align}
where:
\begin{itemize}
    \item $\omega_l$ is the angular frequency,
    \item $a_l (\theta_2,\theta_3)$ is the amplitude of a longitudinal wave traveling in the direction $\mathbf{R}_3(\theta_3)\mathbf{R}_2(\theta_2)\mathbf{k}_l$,
    \item $\mathbf{k}_l$ is a wave vector with $|\mathbf{k}_l| = \omega_l/v_l$ and perpendicular to the (orthogonal) rotation axes of $\mathbf{R}_1$ and $\mathbf{R}_2$.
\end{itemize}

In general, a vector field is uniquely defined by its divergence-free and irrotational components, as stated by the Helmholtz decomposition theorem. In the context of wave mechanics, given time frequencies $\omega_s$ and $\omega_l$, the displacement field $\mathbf{u}$ is entirely described by its curl (divergence-free part) and divergence (irrotational part). Both components admit the integral representations above, showing that the framework developed in earlier sections extends naturally to elastic waves in isotropic solids.

\noindent{\bf Acknowledgment.}  This work
was supported by AFOSR (FA9550-23-1-0093) and the AFOSR-MURI program (FA9550-25-1-0262).

\bibliography{GeneralSolution}
\bibliographystyle{ieeetr}

\appendix
\section{Generalization to evanescent waves or source-free dissipative media}\label{sec:ComplexWaveVector}

In a dissipative medium, electromagnetic waves lose energy as they propagate. This physical behavior can be modeled mathematically by allowing the wave vectors in Maxwell's equation to be complex \cite{Fry_PlaneWO_1927,Burrows_1965}. Such complex wave vectors can also be used to model evanescent waves, which are not really propagating waves, but are solutions to Maxwell's equations. The general class of time-harmonic solutions in a free space was discussed in \Cref{sec:GeneralizationOfTwistedXrays} and is given by \Cref{eq:FinalSolution}. 

\begin{align*}
    \mathbf{E}(\mathbf{x},t) &= e^{-i \omega t} \int_{-\pi}^\pi \int_{-\pi/2}^{\pi/2} A(\theta_2,\theta_3) \mathbf{R}_3 (\theta_3) \mathbf{R}_2 (\theta_2) \mathbf{R}_1 \left(\theta_1(\theta_2,\theta_3)\right) \mathbf{n}_0 e^{i (\mathbf{R}_3 (\theta_3) \mathbf{R}_2 (\theta_2) \mathbf{R}_1\left(\theta_1(\theta_2,\theta_3)\right) \mathbf{k}_0 \cdot \mathbf{x})} \mathrm{d}\theta_2 \mathrm{d}\theta_3
\end{align*}

Here, we have included the rotation $\mathbf{R}_1\left(\theta_1\right)$ inside the exponential for convenience in satisfying the divergence-free condition. In the original formulation, $\mathbf{k}_0 \in \mathbb{R}^3$, but to model dissipation or evanescent waves, we now allow $\mathbf{k}_0 \in \mathbb{C}^3$. Although the original solution applies to free space, it remains valid in a source-free homogeneous medium if we define the wave speed as $c_1 = 1/\sqrt{\mu \epsilon}$, where $\mu$ and $\epsilon$ are the permeability and permittivity of the medium, respectively.
Let $\mathbf{k}_0 = \mathbf{k}_1 + i \mathbf{k}_2$, with $\mathbf{k}_1,\mathbf{k}_2 \in \mathbb{R}^3$. Notice, $\mathbf{k}_0$ uniquely defines $\mathbf{k}_1$ and $\mathbf{k}_2$. Substituting into the Helmholtz equation and divergence-free condition yields:

\begin{align*}
    \mathbf{k}_0\cdot\mathbf{k}_0 &= \mathbf{k}_1\cdot \mathbf{k}_1 - \mathbf{k}_2 \cdot \mathbf{k}_2 + 2i\mathbf{k}_1\cdot\mathbf{k}_2 = k^2,\\
    \mathbf{k}_0 \cdot \mathbf{n}_0 &= \mathbf{k}_1 \cdot \mathbf{n}_0 + i \mathbf{k}_2 \cdot \mathbf{n}_0 = 0.
\end{align*}

Assuming $\mathbf{n}_0 \in \mathbb{R}^3$, the divergence-free condition implies that $\mathbf{n}_0$ is orthogonal to both $\mathbf{k}_1$ and $\mathbf{k}_2$. From the imaginary part of the Helmholtz equation, we also require $\mathbf{k}_1 \cdot \mathbf{k}_2 =0$, so $\mathbf{n}_0$, $\mathbf{k}_1$ and $\mathbf{k}_2$ are mutually orthogonal. Let $\mathbf{k}_1 = k_1 \hat{\mathbf{k}}_1$ and $\mathbf{k}_2=k_2 \hat{\mathbf{k}}_2$ with $\hat{\mathbf{k}}_1$ and $\hat{\mathbf{k}}_2$ unit vectors. Then $\{\hat{\mathbf{k}}_1, \hat{\mathbf{k}}_2, \mathbf{n}_0 \}$ forms an orthonormal basis. From the real part of the Helmholtz equation, we obtain: $$k_1^2-k_2^2 = k^2.$$ We can parameterize this using an angle $\phi \in [0,\pi/2)$ as: $$k_1 = k \sec{\phi}, \quad k_2 = k \tan{\phi}.$$ Thus, the time harmonic solution becomes:  $$\mathbf{E} (\mathbf{x},t) = \mathbf{n}_0 e^{i k (\sec{\phi} \hat{\mathbf{k}}_1 + i \tan{\phi} \hat{\mathbf{k}}_2)\cdot\mathbf{x}} e^{-i \omega t},$$ which satisfies Maxwell's equation in a source-free dissipative medium with wave speed $c_1 = \omega/k$. This leads to the following generalization of \Cref{eq:FinalSolution}: 

\begin{subequations}\label{eq:FinalSolutionDissipativeMedia}
\begin{align}
    \mathbf{E}(\mathbf{x},t) &= e^{-i \omega t} \int_{0}^{\pi/2}\int_{-\pi}^\pi \int_{-\pi/2}^{\pi/2} A(\theta_2,\theta_3,\phi) \mathbf{R}(\theta_2,\theta_3,\phi) \mathbf{n}_0 e^{i k \mathbf{R}(\theta_2,\theta_3,\phi) \hat{\mathbf{k}} (\phi)\cdot \mathbf{x}} \mathrm{d}\theta_2 \mathrm{d}\theta_3 \mathrm{d}\phi\\
    \mathbf{R}(\theta_2,\theta_3,\phi) &= \mathbf{R}_3 (\theta_3) \mathbf{R}_2 (\theta_2) \mathbf{R}_1 \left(\theta_1(\theta_2,\theta_3,\phi)\right)\\
    \hat{\mathbf{k}} (\phi) &= \sec{\phi} \hat{\mathbf{k}}_1 + i \tan{\phi} \hat{\mathbf{k}}_2
\end{align}
\end{subequations}

Similarly, if we want to solve the Helmholtz equation for a scalar-valued function $u$, $\Delta u=-k^2u$, then we can have a scalar version of the above solution as

\begin{subequations}\label{eq:HelmholtzScalarSolutionDissipative}
\begin{align}
    u(\mathbf{x},t) &= e^{-i \omega t} \int_{0}^{\pi/2}\int_{-\pi}^\pi \int_{-\pi/2}^{\pi/2} \hat{u}(\theta_2,\theta_3,\phi)  e^{i k \mathbf{R}(\theta_2,\theta_3,\phi) \hat{\mathbf{k}} (\phi)\cdot \mathbf{x}} \mathrm{d}\theta_2 \mathrm{d}\theta_3 \mathrm{d}\phi\\
    \mathbf{R}(\theta_2,\theta_3,\phi) &= \mathbf{R}_3 (\theta_3) \mathbf{R}_2 (\theta_2) \mathbf{R}_1 \left(\theta_1(\theta_2,\theta_3,\phi)\right)\\
    \hat{\mathbf{k}} (\phi) &= \sec{\phi} \hat{\mathbf{k}}_1 + i \tan{\phi} \hat{\mathbf{k}}_2
\end{align}
\end{subequations}

\begin{remark}
    In the above, we assumed $\mathbf{n}_0 \in \mathbb{R}^3$. A more general formulation allows $\mathbf{n}_0 \in \mathbb{C}^3$, with $\mathbf{n}_0 = \mathbf{n}_1 + i \mathbf{n}_2$, where $\mathbf{n}_1,\mathbf{n}_2 \in \mathbb{R}^3$. The Helmholtz equation imposes no additional constraints on $\mathbf{n}_0$, it only requires $\mathbf{k}_1 \cdot \mathbf{k}_2 =0$ perpendicular. But the divergence-free condition becomes:

    \begin{align*}
        \left(\mathbf{n}_1 + i \mathbf{n}_2 \right) \cdot\left(\mathbf{k}_1 + i \mathbf{k}_2 \right) &= 0,
    \end{align*}
which yields:
    \begin{align*}
       \mathbf{n}_1 \cdot \mathbf{k}_1 = \mathbf{n}_2 \cdot \mathbf{k}_2, \quad & \quad \mathbf{n}_1 \cdot \mathbf{k}_2 = - \mathbf{n}_2 \cdot \mathbf{k}_1.
    \end{align*}

These conditions constrain only the components of $\mathbf{n}_1$ and $\mathbf{n}_2$ in the plane spanned by $\mathbf{k}_1$ and $\mathbf{k}_2$. We can set the perpendicular components to zero, because they do not affect the divergence-free condition and their influence can be absorbed into the amplitude function $A(\theta_2,\theta_3,\phi)$ in Equation \cref{eq:FinalSolutionDissipativeMedia}. Now, using the earlier parameterization: $$\mathbf{k}_1 = k \sec{\phi} \hat{\mathbf{k}}_1, \quad \mathbf{k}_2 = k \tan{\phi} \hat{\mathbf{k}}_2,$$ we obtain:

\begin{align*}
    \sec{\phi}(\mathbf{n}_1 \cdot \hat{\mathbf{k}}_1) = \tan{\phi}(\mathbf{n}_2 \cdot \hat{\mathbf{k}}_2), \quad & \quad \tan{\phi}(\mathbf{n}_1 \cdot \hat{\mathbf{k}}_2) = - \sec{\phi} (\mathbf{n}_2 \cdot \hat{\mathbf{k}}_1).
\end{align*}
Solving these, we find: $$\mathbf{n}_1 = \eta_1 \tan{\phi} \hat{\mathbf{k}}_1 + \eta_2 \sec{\phi}\hat{\mathbf{k}}_2, \quad\ \mathbf{n}_2 = -\eta_2 \tan{\phi}\hat{\mathbf{k}}_1 + \eta_1 \sec{\phi} \hat{\mathbf{k}}_2 $$ for arbitrary $\eta_1,\eta_2\in\mathbb{C}$. Therefore, $$\mathbf{n}_0 = \mathbf{n}_1 + i \mathbf{n}_2 = (\eta_2 + i \eta_1) \left(\sec{\phi} \hat{\mathbf{k}}_2 - i \tan{\phi} \hat{\mathbf{k}}_1 \right).$$ Without loss of generality, we can normalize $\mathbf{n}_0$ such that $\mathbf{n}_0\cdot\mathbf{n}_0 =1$ using $\mathbf{n}_0 = \sec{\phi} \hat{\mathbf{k}}_2 - i \tan{\phi} \hat{\mathbf{k}}_1 = \mathbf{R}_3(\pi/2) \left( \sec{\phi}\hat{\mathbf{k}}_1 + i \tan{\phi} \hat{\mathbf{k}}_2 \right)$. This leads to the final generalized solution: 

\begin{subequations}
\begin{align}\label{eq:GeneralSolutionDissipativeMedia}
    &\mathbf{E}(\mathbf{x},t) = e^{-i \omega t} \int_{0}^{\pi/2}\int_{-\pi}^\pi \int_{-\pi/2}^{\pi/2} \mathbf{R}(\theta_2,\theta_3,\phi) \left[ A(\theta_2,\theta_3,\phi) \mathbf{n}_0 + \tilde{A}(\theta_2,\theta_3,\phi) \mathbf{R}_3(\pi/2) \hat{\mathbf{k}}(\phi) \right] e^{i k \mathbf{R}(\theta_2,\theta_3,\phi) \hat{\mathbf{k}} (\phi) \cdot \mathbf{x}} \mathrm{d}\theta_2 \mathrm{d}\theta_3 \mathrm{d}\phi,\\
    &\mathbf{R}(\theta_2,\theta_3,\phi) = \mathbf{R}_3 (\theta_3) \mathbf{R}_2 (\theta_2) \mathbf{R}_1 \left(\theta_1(\theta_2,\theta_3,\phi)\right),\\
    &\hat{\mathbf{k}} (\phi) = \sec{\phi} \hat{\mathbf{k}}_1 + i \tan{\phi} \hat{\mathbf{k}}_2.
\end{align}
\end{subequations}
\end{remark}

\section{Maximum magnitude by linear superposition of complex vectors}\label{sec:MaxMag}

Consider a set of $n$ two-dimensional complex vectors $\left(\mathbf{v}_j\right)_{j=1}^n$. The goal is to determine complex coefficients $\left(\alpha_j\right)_{j=1}^n$, with $\left| \alpha_j \right| \leq 1$, that maximize the magnitude of their linear combination: $$\left| \sum_j \alpha_j \mathbf{v}_j \right| = \sup_{\tilde{\alpha}_j : \left| \tilde{\alpha}_j \right| \leq 1} \left| \sum \tilde{\alpha}_j \mathbf{v}_j \right|.$$ 

Let $\mathbf{v}_{opt} = \sum_j \alpha_j \mathbf{v}_j$ denote the optimal resultant vector. We assume that all vectors $\mathbf{v}_j$ are non-zero, without loss of generality. Here, $|\mathbf{v}_j|$ is the complex norm, defined as $\sqrt{\langle \mathbf{v}_j, \mathbf{v}_j \rangle}$, and the inner product in complex vector space is given as $\langle \mathbf{v}_j, \mathbf{v}_k \rangle = \mathbf{v}_j \cdot \mathbf{v}_k^*$, where $(\cdot)^*$ denotes complex conjugation. Note that $\langle \mathbf{v}_j, \mathbf{v}_k \rangle = \langle \mathbf{v}_k, \mathbf{v}_j \rangle^*$. Although a closed-form solution is not known, we can establish the following key properties:

\begin{claim} $\Re\langle \mathbf{v}_{opt} , \mathbf{v}_j \rangle \neq 0$ for all $j=1,2,\dots,n$. (Remember,  $\mathbf{v}_{opt} = \sum_j \alpha_j \mathbf{v}_j$ is the optimum resultant vector).
\end{claim}

\begin{proof}
    Assume, for contradiction, that $\Re\langle \mathbf{v}_{opt} , \mathbf{v}_k \rangle = 0$ for some $k$, and let $\alpha_k =\gamma_k e^{i\phi_k}$ be the corresponding coefficient with $0\leq\gamma_k\leq1$. Consider two cases:
    If $\gamma_k<1$, then increasing its magnitude to 1 by adding $(1-\gamma_k)e^{i \phi_k}\mathbf{v}_k $ yields: $$ |\mathbf{v}_{opt} + (1-\gamma_k)e^{i \phi_k}\mathbf{v}_k|^2 =  |\mathbf{v}_{opt}|^2 + (1-\gamma_k)^2 |\mathbf{v}_k|^2, $$ which contradicts optimality. 
    If $\gamma_k=1$, then subtracting a small multiple $\delta e^{i\phi_k} \mathbf{v}_k$ with $0<\delta<1$ yields: $$ |\mathbf{v}_{opt} - \delta e^{i \phi_k}\mathbf{v}_k|^2 =  |\mathbf{v}_{opt}|^2 + \delta^2 |\mathbf{v}_k|^2, $$ again contradicting optimality. In both cases, the assumption leads to a contradiction. Hence, the claim holds.
\end{proof}

\begin{claim} $\Re\langle \mathbf{v}_{opt} , \alpha_j \mathbf{v}_j \rangle > 0$ for all $j=1,2,\dots,n$.
\end{claim}

\begin{proof}
Assume the contrary: that for some $k$, $\Re \langle \mathbf{v}_{opt},\alpha_k \mathbf{v}_k\rangle < 0$. Then consider the vector $\mathbf{v}_{opt}' = \mathbf{v}_{opt}-2\alpha_k \mathbf{v}_k = \sum_j \tilde{\alpha}_j \mathbf{v}_j$, which still satisfies the constraint $|\tilde{\alpha}_j|\leq1$. A direct computation shows: $$|\mathbf{v}_{opt}'|^2 > |\mathbf{v}_{opt}|^2,$$ contradicting the optimality of $\mathbf{v}_{opt}$. Hence the claim holds.
\end{proof}

\begin{claim}
$|\alpha_j|=1$  for all $j=1,2,\dots,n$.
\end{claim}

\begin{proof}
Suppose $|\alpha_k|=r_k < 1$ for some $k$. Then increasing its magnitude to 1 (while preserving its phase) increases the real part of the inner product with $\mathbf{v}_{opt}$, and thus the norm of the sum. This contradicts optimality, so all coefficients must lie on the unit circle: $\alpha_j=e^{i\psi_j}$.
\end{proof}

\begin{claim}
$$\Re\langle\mathbf{v}_{opt} , e^{i \psi_j} \mathbf{v}_j\rangle = \sup_{\tilde{\psi}_j} \Re\langle\mathbf{v}_{opt} , e^{i \tilde{\psi}_j} \mathbf{v}_j\rangle \quad \text{ for all } j.$$
\end{claim}

\begin{proof}
If this were not true for some $j$, we could increase the real part of the inner product by adjusting the phase $\psi_j$, thereby increasing the norm of the sum—again contradicting optimality.
\end{proof}

\textbf{Calculating the coefficients}: 

Let $\langle \mathbf{v}_{opt}, \mathbf{v}_k \rangle = \gamma_k e^{i \beta_k}$, with $\gamma_k >0$.

Then, from the last claim, the corresponding optimal coefficient is: $$\alpha_k = e^{i\beta_k} = \frac{\langle \mathbf{v}_{opt}, \mathbf{v}_k \rangle}{|\langle \mathbf{v}_{opt}, \mathbf{v}_k \rangle|}.$$

Thus, the optimal vector satisfies the implicit equation:

\begin{align*}
    \mathbf{v}_{opt} = \sum_{j=1}^{n} \frac{\left< \mathbf{v}_{opt}, \mathbf{v}_j \right>}{\left| \left< \mathbf{v}_{opt}, \mathbf{v}_j \right> \right|} \mathbf{v}_j.
\end{align*} 

This equation can be solved numerically by searching over all unit vectors $\hat{\mathbf{v}}\in \mathbb{C}^2$ and evaluating the right-hand side. The direction that maximizes the resulting magnitude corresponds to $\mathbf{v}_{opt}$.

To implement the numerical method, we must assume a unit vector in complex two-dimensional space. A general unit vector in $\mathbb{C}^2$ can be given as: $$\mathbf{v} = \begin{pmatrix} e^{i \tilde{\theta}_2}\cos{\theta} \\ e^{i \tilde{\theta}_1} \sin{\theta} \end{pmatrix},$$ where $\theta \in [0,\pi/2]$ and  $\tilde{\theta_1},\tilde{\theta_2}\in[0,2\pi]$. However, the solution to the optimization problem is not unique: if $\mathbf{v}_{opt}$ is a solution, then so is $e^{i \tilde{\theta}} \mathbf{v}_{opt}$ for any global phase $\tilde{\theta}\in \mathbb{R}$. To reduce the computational cost, we can fix the global phase without loss of generality. This allow us to represent the unit vector more simply as $$\mathbf{v} = \begin{pmatrix} \cos{\theta} \\ e^{i \theta_1} \sin{\theta} \end{pmatrix},$$ with $\theta \in [0,\pi/2]$ and  $\theta_1\in[0,2\pi]$. This simplification significantly reduces the dimensionality of the search space during numerical optimization.

\textbf{Special case of real vectors : }  If all vectors $\left( \mathbf{v}_j \right)_{j=1}^{n}$ are real, the optimal coefficients simplify considerably. In this case:$$\frac{\left< \mathbf{v}_{opt}, \mathbf{v}_j \right>}{\left| \left< \mathbf{v}_{opt}, \mathbf{v}_j \right> \right|} = \text{sgn} \left( \left< \mathbf{v}_{opt}, \mathbf{v}_j \right> \right) = \text{sgn} \left( \cos(\theta_j) \right), $$ where $\theta_j$ is the angle between $\mathbf{v}_{opt}$ and $\mathbf{v}_j$. Thus, each coefficient is either $+1$ or $-1$, depending on whether the angle is acute or obtuse. This result aligns with the geometric intuition: vectors forming acute angles with the resultant vector contribute constructively (positive sign), while those forming obtuse angles contribute destructively (negative sign).

\end{document}